%% file: main.tex
\documentclass[11pt,a4paper]{article}

\usepackage[utf8]{inputenc}
\usepackage[T1]{fontenc}
\usepackage[margin=1in]{geometry}
\usepackage{amsmath,amssymb,amsthm}
\usepackage{graphicx}
\usepackage{booktabs}
\usepackage{subcaption}
\usepackage{multirow}
\usepackage{listings}
\usepackage{algorithm}
\usepackage{xcolor}
\usepackage{tikz}
\usepackage{microtype}
\usepackage[numbers,sort&compress]{natbib}
\usepackage[hidelinks,breaklinks=true]{hyperref}

\usetikzlibrary{arrows.meta,positioning,fit,backgrounds,calc}

\graphicspath{{figures/}}

\definecolor{codebg}{RGB}{247,248,250}
\definecolor{codekw}{RGB}{20,80,160}
\definecolor{codecm}{RGB}{110,120,130}
\lstdefinestyle{xg}{
  backgroundcolor=\color{codebg},
  basicstyle=\ttfamily\footnotesize,
  keywordstyle=\color{codekw}\bfseries,
  commentstyle=\color{codecm}\itshape,
  breaklines=true,
  columns=fullflexible,
  frame=single,
  rulecolor=\color{codecm!40},
  framesep=4pt,
  xleftmargin=6pt,
  showstringspaces=false,
  captionpos=b,
}
\newtheorem{proposition}{Proposition}
\newtheorem{definition}{Definition}
\theoremstyle{remark}
\newtheorem{remark}{Remark}

\newcommand{\EMB}{\textsc{embedded}}
\newcommand{\code}[1]{\texttt{\small #1}}
\newcommand{\sumw}{\ensuremath{S_w}}
\newcommand{\sumwv}{\ensuremath{S_{wv}}}

\title{\bfseries Hidden relationships in a document-derived property graph:\\
top-$k$ chunk embeddings and inverse-distance weighting\\
over a dynamically evolving ontology}

\author{Bilge Kaan Karamete\thanks{Corresponding author:
  \texttt{bkaramete@babelstreet.com}}, PhD\ \ and\ \ Hunter Casten\\[2pt]
  \small Babel Street\\
  \small 1900 Reston Metro Plaza, Suite 950, Reston, VA 20190, USA}

\date{}

\begin{document}
\maketitle

\begin{abstract}
A large language model asked to build a knowledge graph from text will assert
only what a sentence states. Over a corpus, that faithfulness is also the
method's ceiling: three paragraphs about one family can produce two disconnected
components, because no single sentence ever names the tie between them. This
article describes a purely \emph{additive} second pass that recovers those ties
without touching what the extractor found. Each document is split into paragraph
chunks and embedded once. A top-$k$ nearest-neighbour query --- originating only
in the new document's chunks, but searching every chunk in the graph --- yields
chunk pairs, which expand through a chunk-to-entity membership map into candidate
node pairs. Each pair is then scored by Shepard inverse-distance
weighting~over every chunk pair its endpoints co-occur in. Because the vectors
are $L_2$-normalised at write, $\lVert a-b\rVert^2 = 2(1-\cos)$, which makes the
weight $w = 1/(\varepsilon + 1 - \cos)$ exactly Shepard's inverse-\emph{squared}
distance and the value term a rescaled chord distance,
$v = 1 - \lVert a-b \rVert / 2$. We show that the affine alternative
$v = (1+\cos)/2$ is unusable behind a $k$-NN gate --- the gate alone confines it
to a band of width $\beta/2$, silently disabling the write threshold --- and we
derive the resulting constraint $\theta > 1-\sqrt{\beta/2}$. The per-pair
accumulators $(\sum wv, \sum w, \text{support})$ form a commutative monoid and
are persisted \emph{un-gated}, so the pass is order-independent, needs no
recomputation as the corpus grows, and lets a node pair that misses the threshold
today cross it after the next document arrives. The implementation is
engine-neutral: the same weighted \EMB{} edges go into FalkorDB, Kinetica,
ArangoDB or Neo4j through one adapter contract. We report the measurements that
shaped the design --- $768$-dimensional embeddings agreeing with a
$3072$-dimensional reference on $92\%$ of edges and $240$-dimensional on $72\%$
at under a thirteenth of the storage, and a top-$k$ formulation $25\times$ faster
than the window-function form --- and we are explicit about what is \emph{not}
measured: precision and recall against a labelled set of hidden relationships,
for which no ground truth exists in this corpus.
\end{abstract}

\medskip
\noindent\textbf{Keywords:} knowledge graph construction; property graph; link
prediction; inverse distance weighting; text embeddings; $k$-nearest neighbours;
incremental ontology; large language models.

\input{sections/01-introduction}
\input{sections/02-related-work}
\input{sections/03-pipeline}
\input{sections/04-embedding}
\input{sections/05-idw}
\input{sections/06-incremental}
\input{sections/07-implementation}
\input{sections/08-results}
\input{sections/09-scale}
\input{sections/10-limitations}
\input{sections/11-conclusion}

\section*{Acknowledgements}
The engineering record this article draws on --- every measurement quoted here
was taken against the running system --- was kept as part of the \emph{xgraph}
workbench at Babel Street.

\bibliographystyle{plainnat}
\bibliography{refs}

\end{document}

%% file: sections/01-introduction.tex
\section{Introduction}
\label{sec:intro}

Building a property graph from unstructured text is now a routine application of
large language models: a document is split into passages, each passage is handed
to a model under a constrained output schema, and the entities and relations it
returns are folded into a canonical ontology and written into a graph engine.
The method's great virtue is that it is \emph{faithful} --- every edge it
produces corresponds to something a sentence actually said, and a well-prompted
extractor will refuse to invent the rest.

That virtue is also its ceiling. Consider the three short paragraphs in
Figure~\ref{fig:extract}, drawn from the running example used throughout this
article. They concern one family. A structural extractor reads them correctly
and produces four edges --- \code{Gulgun}\,$\to$\,\code{FSI},
\code{Kaan}\,$\to$\,\code{BabelStreet}, \code{Tan}\,$\to$\,\code{Bloomberg},
\code{Tan}\,$\to$\,\code{Kaan} --- and every one of them is a sentence someone
wrote. It nonetheless leaves \emph{two disconnected components}, because no
sentence in the corpus ever states the tie between \code{Gulgun} and
\code{Tan}, or between \code{FSI} and \code{BabelStreet}. The information that
these paragraphs belong together is present in the corpus; it is simply not
present in any one sentence, and so it is invisible to a method that only
transcribes sentences.

\begin{figure}[t]
  \centering
  \includegraphics[width=\textwidth]{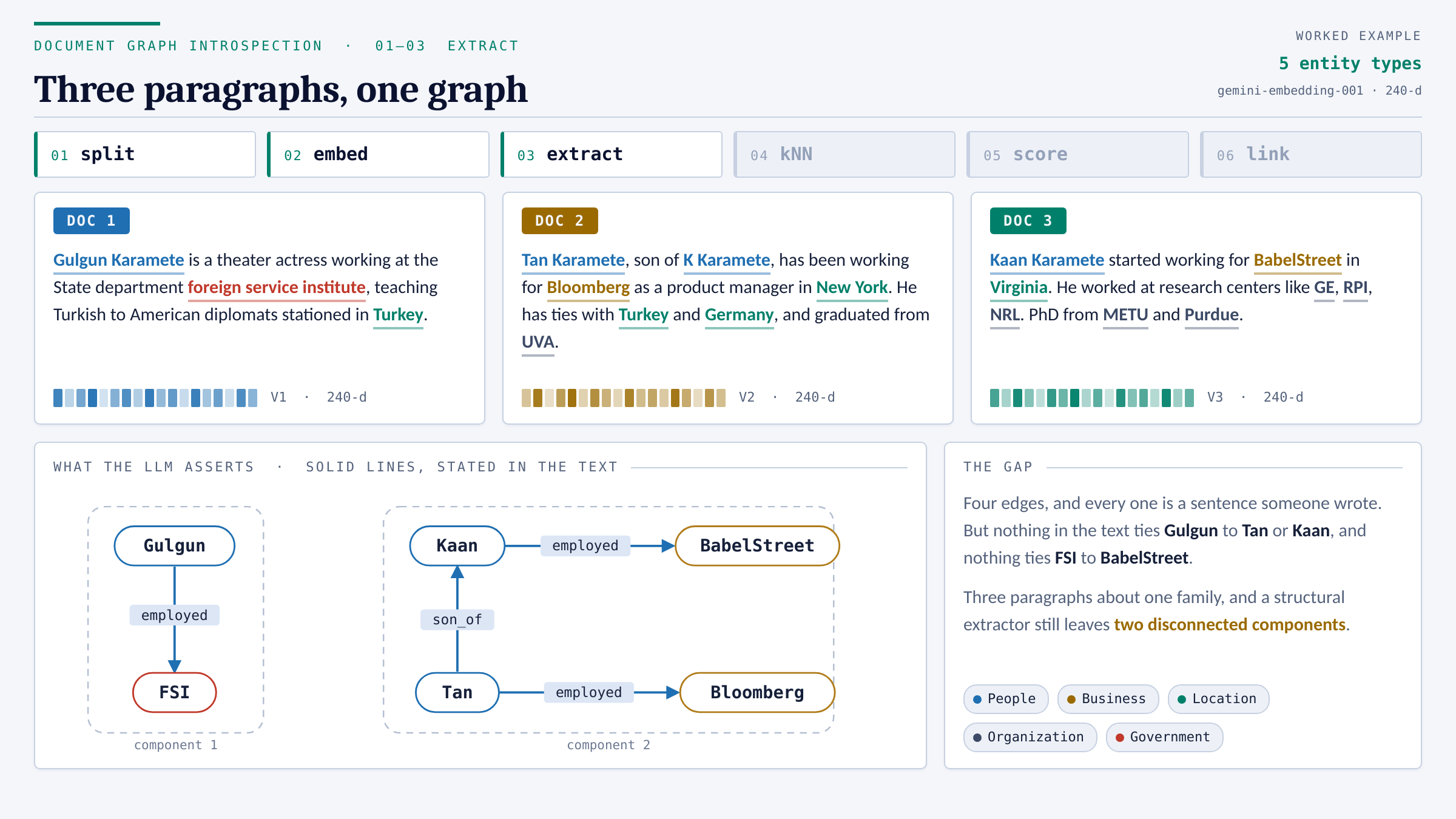}
  \caption{The running example. Three paragraphs, embedded at $240$ dimensions
  and extracted independently, yield a correct but disconnected graph: four
  stated edges, two components. Nothing in the text ties \code{Gulgun} to
  \code{Tan} or \code{FSI} to \code{BabelStreet}. The colour strips below each
  paragraph are its embedding, mean-pooled to $64$ buckets for display
  (Section~\ref{sec:impl}).}
  \label{fig:extract}
\end{figure}

The obvious repair is to add edges the text did not state, and the obvious
danger is that such edges are no longer facts. Our position is that the two can
be reconciled by keeping the derived layer strictly separate and strictly
additive: it writes exactly one new edge label, \EMB{}, carrying a numeric
weight and a support count; it can never modify or delete anything extraction
produced; and it is hidden from the natural-language query path, so that a
cosine artefact is never reported to a user as an extracted fact.

\subsection{The approach in one paragraph}

Each document is split into paragraph-sized chunks and embedded once. A top-$k$
nearest-neighbour query is issued whose \emph{query} side is restricted to the
new document's chunks but whose \emph{search} side is every chunk in the graph;
this is what links an arriving document to the corpus already ingested without
ever re-comparing the corpus against itself. The resulting chunk pairs, together
with same-chunk co-occurrence entered as a self-pair at $\cos = 1$, are expanded
through a chunk-to-entity membership map into candidate node pairs. Each node
pair is then scored by Shepard inverse-distance weighting~\citep{shepard1968}
over all chunk pairs in which its two endpoints co-occur, and pairs whose weight
clears a threshold are written back as weighted \EMB{} edges
(Figure~\ref{fig:connect}). The seven stages --- \emph{embed, persist, $k$-NN,
expand, accumulate, select, link} --- are the organising structure of this
article. That threshold is the layer's only admission rule, and it is a real
one: Figure~\ref{fig:example-theta} draws these same three paragraphs at three
settings of it, and at the third the two components are two again.

\begin{figure}[t]
  \centering
  \includegraphics[width=\textwidth]{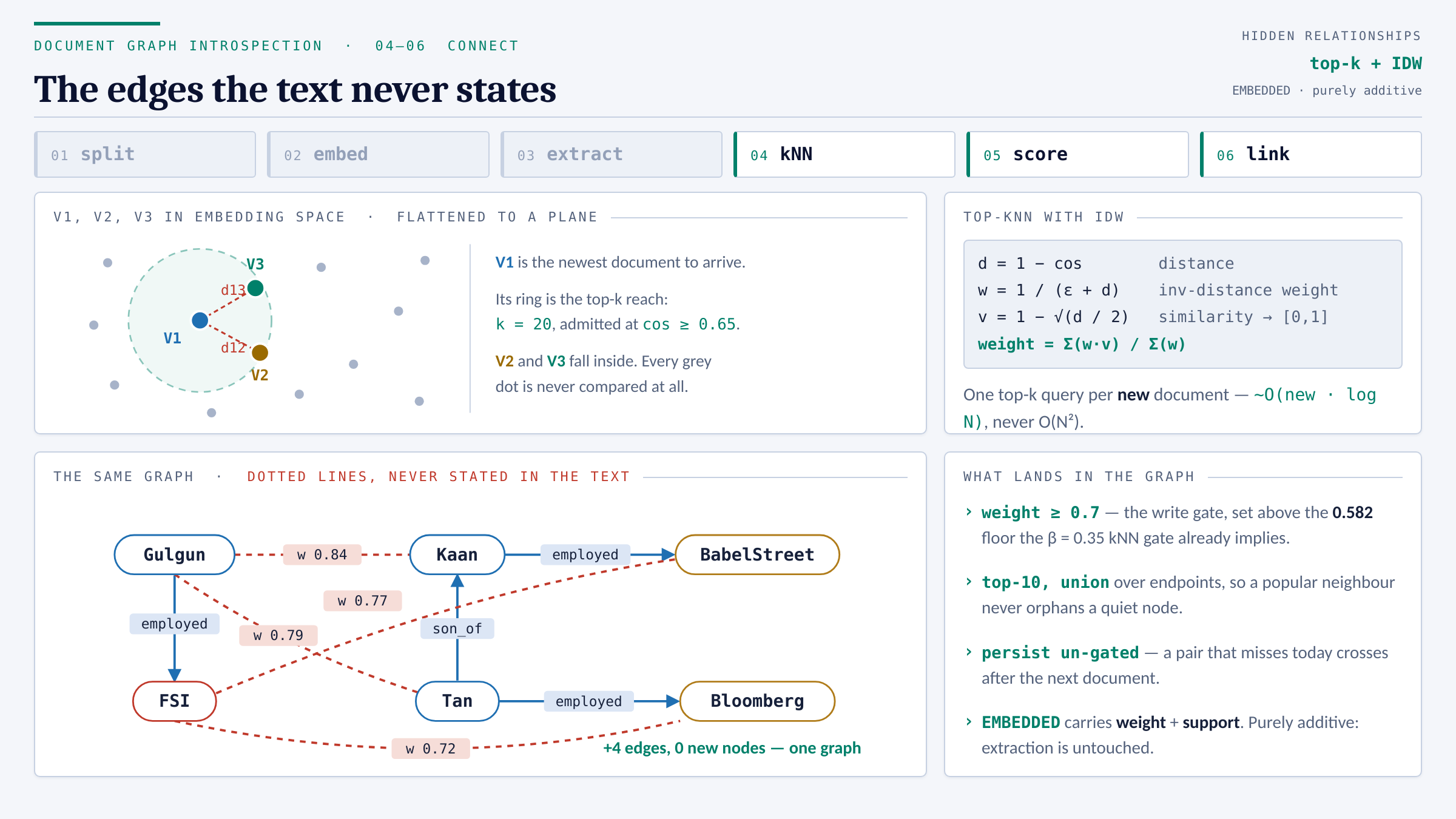}
  \caption{The same graph after the similarity pass. Four weighted \EMB{} edges
  (dotted) join the two components; no new nodes are created. The top-$k$ ring
  on the left is the reach of one document's chunk: neighbours outside it are
  never compared at all, which is what keeps the per-document cost sublinear in
  corpus size rather than quadratic.}
  \label{fig:connect}
\end{figure}

\subsection{Why inverse-distance weighting}

The scoring problem is an interpolation problem in disguise. A node pair is
observed at several \emph{sites} --- the chunk pairs its endpoints co-occur in
--- each site carrying both a similarity value and a distance, and we want one
number that respects the near observations more than the far ones. That is
precisely the setting Shepard's method was written for, and it brings two
properties that matter more here than accuracy does. First, it is a ratio of
two sums, so the accumulators are additive and a new observation is folded in
by addition alone. Second, it degrades gracefully with a single very close
observation, which is what we want when two entities share a paragraph outright.

The arithmetic simplifies considerably because the vectors are $L_2$-normalised
at write time. On unit vectors $\lVert a-b\rVert^2 = 2(1-\cos)$, so the
inverse-distance weight $w = 1/(\varepsilon + 1 - \cos)$ is \emph{exactly}
Shepard's inverse-squared-distance weight, with the factor of two cancelling in
the ratio, and the value term admits a clean geometric reading as a rescaled
chord distance. Section~\ref{sec:idw} develops this and shows why the obvious
affine alternative, $v = (1+\cos)/2$, cannot be used behind a $k$-NN gate.

\subsection{Why the ontology is dynamic, and why that constrains the design}

The ontology this pass writes into is not fixed in advance. Entity and relation
type names proposed by the extractor are folded to canonical forms, aliases are
learned as they are encountered, and facets --- dimensions \emph{of} a type,
such as \code{Engineer} on a \code{Person} --- accumulate separately from
structural types. A corpus that grows therefore changes the ontology, not merely
the instance data.

A derived layer over a moving target has to be built so that arrival order does
not matter and so that nothing already computed has to be revisited. Both fall
out of one decision: the accumulators are persisted \emph{un-gated}. The write
threshold is applied only to what is emitted as an edge, never to what is
recorded, so a pair whose evidence is currently too thin keeps accumulating and
crosses the threshold on some later document --- with no recomputation, and with
the same final weight regardless of the order the documents arrived in.
Section~\ref{sec:incremental} states this precisely: the accumulator triple is a
commutative monoid under addition, and the pass is a monoid homomorphism from
document batches into it.

\subsection{Contributions}

\begin{enumerate}
\item An \textbf{additive similarity layer} for LLM-constructed property graphs.
      It recovers relationships the text never states, as one new edge label
      carrying a weight and a support count, and it has no path by which it
      could alter what extraction found (Sections~\ref{sec:pipeline}
      and~\ref{sec:impl}).

\item A \textbf{corrected IDW formulation for cosine neighbourhoods}. Behind a
      $k$-NN gate admitting $\cos \geq 1-\beta$, the affine value term
      $v = (1+\cos)/2$ is confined to a band of width $\beta/2$ and renders any
      write threshold inoperative. We replace it with the rescaled chord distance
      $v = 1-\sqrt{(1-\cos)/2}$, derive the resulting constraint
      $\theta > 1-\sqrt{\beta/2}$, and document the near-miss repair --- Shepard
      $p=1$, which flattens the neighbour spread from $350\times$ to $19\times$
      (Section~\ref{sec:idw}).

\item An \textbf{order-independent incremental construction}. The per-pair
      accumulators are persisted before thresholding, which makes arrival order
      irrelevant and removes recomputation entirely; the per-document cost is one
      top-$k$ query against the corpus, never an all-pairs pass
      (Section~\ref{sec:incremental}).

\item A \textbf{dimension study with a storage argument}. Because
      \code{gemini-embedding-001} is a Matryoshka model~\citep{kusupati2022},
      requesting $N$ dimensions is byte-identical to truncating the full vector,
      so width is a free choice at request time. Against a $3072$-dimensional
      reference, $768$ dimensions agree on $92\%$ of emitted edges and $240$ on
      $72\%$, at $960$ bytes per chunk against $12{,}288$
      (Section~\ref{sec:embedding}).

\item An \textbf{engine-neutral realisation}. The same pass writes the same
      \EMB{} edges into FalkorDB, Kinetica, ArangoDB and Neo4j through one
      adapter contract, and declines --- by name, not by exception --- on an
      engine whose label model admits no additive element write
      (Section~\ref{sec:impl}), continuing the vendor-neutral line of
      \citet{karamete2022}.

\item \textbf{Two performance results} that were design-determining rather than
      incidental: the top-$k$ self-join as a \code{max\_by} aggregate rather than
      a window function is $25\times$ faster at $768$ dimensions over $100$k
      rows; and folding the accumulators in one multi-row \code{VALUES} upsert
      rather than per-row statements reduces a $19{,}900$-pair merge from
      $82.3$\,s to a single statement (Section~\ref{sec:results}).
\end{enumerate}

\subsection{What this article does not claim}

We do not report precision and recall for the recovered relationships. No
labelled ground truth of hidden relationships exists for this corpus, and
constructing one is a separate undertaking; the agreement figures in
Section~\ref{sec:embedding} compare one configuration against another, which is
a statement about stability under dimensionality reduction and not about
correctness. Section~\ref{sec:limitations} sets out this and the other
boundaries of the work explicitly, including the cases where the pass is
deliberately inapplicable.

\subsection{Organisation}

Section~\ref{sec:related} places the work among LLM-based graph construction,
link prediction and spatial interpolation. Section~\ref{sec:pipeline} describes
the document-to-graph pipeline the pass attaches to and the membership map it
consumes. Section~\ref{sec:embedding} covers the embedding model, the two
supported widths and the per-graph pinning that width and model choice require.
Section~\ref{sec:idw} is the core: the IDW formulation, its derivation on unit
vectors, and the two traps. Section~\ref{sec:incremental} treats the incremental
and order-independence properties, and the sense in which the ontology evolves.
Section~\ref{sec:impl} gives the implementation, including the query
formulations that made it fast. Section~\ref{sec:results} reports measurements,
Section~\ref{sec:limitations} the limitations, and Section~\ref{sec:conclusion}
concludes.

%% file: sections/02-related-work.tex
\section{Related work}
\label{sec:related}

The pass described here sits at the meeting point of four literatures that
rarely cite one another: knowledge-graph construction from text, link prediction
on graphs, spatial interpolation, and approximate nearest-neighbour search. We
take each in turn, and in each case try to say precisely what is borrowed and
what is different.

\subsection{Knowledge graphs from text, and the RAG lineage}

Knowledge-graph construction from unstructured text has a long history predating
language models~\citep{hogan2021,ji2022}, but the current practice --- prompt a
model per passage under a constrained output schema, then reconcile the results
--- is recent. Retrieval-augmented generation~\citep{lewis2020} established the
chunk-and-embed substrate that this practice runs on, and GraphRAG
\citep{edge2024} showed that imposing graph structure over those chunks improves
query-focused summarisation, using community detection over an
extraction-derived graph to build hierarchical summaries.

Our setting differs in what the derived structure is \emph{for}. GraphRAG builds
graph structure in service of generating better answers; the graph is
intermediate. Here the property graph is the product --- it is queried directly
in Cypher or GQL, browsed, and joined against relational warehouse tables --- so
a derived edge must be first-class enough to traverse and honest enough to be
distinguishable from an extracted one. That is why \EMB{} is a distinct label
with its own weight and support attributes rather than an unmarked edge, and why
it is withheld from the schema shown to the natural-language query path
(Section~\ref{sec:impl}).

The reconciliation half --- folding \code{Employer}, \code{employer} and
\code{works\_for} into one canonical relation type --- is entity and predicate
resolution under a different name. What is unusual in our setting is that
canonicalisation is not a batch step over a closed corpus but an online one:
the ontology is a live object that gains types, aliases and facets as documents
arrive, which is the condition Section~\ref{sec:incremental} is written against.

\subsection{Link prediction}

Predicting edges absent from an observed graph is a mature
field~\citep{libennowell2007,lu2011}. The dominant modern approaches are
\emph{structural}: they embed the graph itself and score pairs in that space.
Knowledge-graph embedding methods such as TransE~\citep{bordes2013} and
ComplEx~\citep{trouillon2016} learn entity and relation vectors from the
observed triples; random-walk methods such as node2vec~\citep{grover2016} learn
node representations from graph neighbourhoods; and graph neural network methods
such as SEAL~\citep{zhang2018} learn from enclosing subgraphs. Closest to the
present work, \citet{karamete2025embedding} compute node vectors \emph{ad hoc}
inside the graph engine, without a training phase, and score candidate pairs by
an inverse-distance weighting over the resulting neighbourhoods --- the
formulation this article adopts, and whose value term it revises in
Section~\ref{sec:idw}.

Our signal is not structural at all. It is the \emph{provenance} signal --- which
paragraphs an entity was mentioned in --- scored in the embedding space of the
\emph{text}, not of the graph. This has one clear disadvantage and two
advantages. The disadvantage is that we cannot predict an edge between entities
whose supporting text is dissimilar, however suggestive the graph topology may
be; structural methods can, and the two signals are complementary rather than
competing. The advantages are that no training is required, so the layer works
on a graph one document old, and that the evidence for every derived edge is a
concrete set of paragraph pairs, which makes the support count meaningful and
the edge auditable. A structural embedding gives a score; a chunk pair gives a
citation.

\subsection{Inverse-distance weighting}

Shepard's method~\citep{shepard1968} interpolates a value at an unsampled
location as a weighted mean of observed values, weighting each observation by an
inverse power of its distance. It is the standard workhorse of spatial
interpolation and is normally discussed with $p=2$ in two or three geometric
dimensions.

We use it in a high-dimensional cosine space, and the transplant is not
cosmetic. Three points deserve emphasis. First, on $L_2$-normalised vectors the
relation $\lVert a-b\rVert^2 = 2(1-\cos)$ makes the natural cosine-based weight
$1/(1-\cos)$ coincide exactly with Shepard $p=2$, so the choice of exponent is
not a free parameter here but a consequence of normalisation
(Section~\ref{sec:idw}). Second, the classical method interpolates a value that
is given independently of the distances, whereas our value term is a
\emph{function} of the same distance that produces the weight --- which is why
the choice of that function interacts with the neighbourhood gate in the way
Section~\ref{sec:idw} analyses, an interaction with no counterpart in the
spatial setting. Third, the singularity at zero distance, normally an annoyance
handled by an $\varepsilon$ or by returning the observed value exactly, is here
a deliberate feature: two entities in the same paragraph enter at $\cos = 1$ and
dominate, because sharing a paragraph is much stronger evidence than being
mentioned in similar paragraphs.

\subsection{Text embeddings and nearest-neighbour search}

Dense sentence and passage embeddings~\citep{reimers2019}, benchmarked at scale
by MTEB~\citep{muennighoff2023}, are the input to the pass. We use
\code{gemini-embedding-001} \citep{lee2025,geminiteam2023}, whose relevant
property for this work is that it is trained with Matryoshka representation
learning~\citep{kusupati2022}: prefixes of the full vector are themselves valid
embeddings, so requesting $N$ dimensions is byte-identical to truncating to the
first $N$ components. Width therefore becomes a storage decision made at
request time rather than a modelling decision made at training time, which is
what allows the two-width design of Section~\ref{sec:embedding}.

For neighbour search, HNSW~\citep{malkov2020} and the Faiss
library~\citep{johnson2019,douze2024} are the standard tools at scale. We use
neither as the default, and reach for the first only past a corpus-size
threshold. The query is by default an \emph{exact} top-$k$ self-join inside the
analytical store~\citep{raasveldt2019} (Algorithm~\ref{alg:master}, step~8),
which is affordable because only one arriving document's chunks are on the query
side; past roughly $10^5$ stored vectors, where the scan stops being cheaper than
the index, the store's own HNSW index supplies candidates and the exact metric
still supplies the scores. That the approximation is confined to candidate
selection is the whole of what makes it acceptable, and
Section~\ref{subsec:cost} measures the residual at the level that matters --- the
edges finally written, not the neighbours retrieved.
Remark~\ref{rem:notasearch} draws the distinction that survives either route.

The other classical way to turn a vector set into candidate pairs is to
\emph{partition} it rather than to query it --- $k$-means and its spherical
variant on unit vectors~\citep{dhillon2001}, the self-organising map
\citep{kohonen1990}, or any of the clustering family surveyed
by~\citet{jain2010}. We measured that alternative rather than dismissing it,
both as a scoring rule and as an index, and neither form is adopted. The negative
result generalises further than we expected: partitioning fails here on the raw
coordinates, on a wavelet basis, and on the centred and re-isotropised vectors
alike, which is what eventually motivated a \emph{graph} index rather than a
better partition. Section~\ref{subsec:cost} gives the measurements.

\subsection{Property graphs and engine neutrality}

The graph is expressed in the property-graph model and queried in Cypher
\citep{francis2018}, or in GQL and SQL/PGQ~\citep{deutsch2022}, whose
foundations are surveyed by \citet{angles2017}; \citet{robinson2015} is the
standard practical treatment. A recurring obstacle to writing a derived layer once and running it
everywhere is that engines disagree on where a label \emph{lives}: FalkorDB
\citep{cailliau2019}, Neo4j and ArangoDB treat a label as a value on the
element, so a new edge label is an ordinary write, whereas warehouse-backed
property graphs in the Dremel lineage~\citep{melnik2010} treat a label as part
of the schema, so introducing one is a DDL change and there is no additive
element write at all. Our adapter contract makes that difference explicit
rather than papering over it: the pass names the engine and declines, instead of
failing obscurely (Section~\ref{sec:impl}).

Where a label \emph{is} a value, what it costs to filter on one is a data
structure question rather than a query-language one, and
\citet{karamete2023labels} treat it as such at billion-edge scale. That cost is
not incidental here. A pass that adds one dominant edge label makes label-scoped
traversal the primary way to read the graph at all --- separating the derived
layer from the extracted one is a label predicate, and
Section~\ref{subsec:derived} shows it is the difference between a legible graph
and an opaque one.

The vendor-neutral framing, and the concern with keeping graph traversal and
relational analytics in one system rather than two, continues the line of
\citet{karamete2022}.

%% file: sections/03-pipeline.tex
\section{From documents to a property graph}
\label{sec:pipeline}

This section describes the extraction pipeline the similarity pass attaches to.
It is deliberately conventional; what matters for the rest of the article is the
three artefacts it leaves behind --- the chunk sequence, the canonical ontology,
and the chunk-to-entity membership map --- and the additive write seam through
which anything downstream must go.

\subsection{Chunking}

A document is split on paragraph boundaries. Paragraphs, rather than a fixed
token window, are used because they are the unit the extractor is prompted
against and therefore the unit whose membership map is meaningful: an entity is
recorded as a member of the paragraph the model saw when it named that entity.

Two details are load-bearing. The number of chunks per document is capped, and
the cap is \emph{reported} rather than applied silently --- the pipeline returns
a truncation flag alongside the chunk list, so a document that lost its tail is
visible as such rather than appearing to have been processed in full. And chunk
extraction calls are issued concurrently but their results are stored by index,
so the downstream merge is order-deterministic even though completion order is
not.

\subsection{Extraction}

Each chunk is sent to a language model under a constrained JSON output schema
requesting a list of entities and a list of relations. The schema enforces one
invariant that later stages depend on: a relation's \code{source} and
\code{target} must each equal the \code{name} of an entity returned \emph{from
that same chunk}. This keeps every asserted edge local to a passage --- an edge
is a claim about one paragraph, and its provenance is unambiguous --- and it is
what makes the membership map of Section~\ref{subsec:membership} well defined.

Extraction quality is treated as a build-time concern rather than a query-time
one, so extraction and the folding step below run on a stronger model tier than
the interactive question-answering path does.

\subsection{Folding into a canonical ontology}
\label{subsec:folding}

A language model asked for an entity type will return \code{Company} on one
paragraph and \code{Organization} on the next. The folding stage rewrites
proposed type labels to canonical forms and learns aliases as it goes, along two
axes --- \code{EntityType} and \code{RelationType} --- with the decision itself
delegated to the model when a proposed name is not already a known alias.

The ontology that results is a live object, and it grows in three distinct ways
as documents arrive: new canonical types appear, new aliases are attached to
existing canonicals, and \emph{facets} accumulate. A facet is a dimension of a
type rather than a type in its own right --- \code{Engineer} classifies a
\code{Person} --- and it is kept separate from the structural type on purpose. A
node's label vector is $[\textit{structural type}, \textit{facets}\ldots]$, and
folding a facet into the structural position would fan every relation edge out
over the label combinations and push the per-type node shares past $100\%$. The
distinction matters to the similarity pass only indirectly, but it matters to
the claim that the ontology \emph{evolves}: the schema a query sees after ten
documents is not the schema it saw after one.

\subsection{The membership map}
\label{subsec:membership}

The artefact the similarity pass actually consumes is the map
\[
  \mu : (\textit{doc\_uri}, \textit{chunk\_index}) \;\longmapsto\;
  \{\,\textit{node id}\,\},
\]
recording which entities were extracted from which chunk. It is written at
extraction time, one row per (chunk, entity) pair, and it is the only bridge
between the embedding space --- which knows about paragraphs --- and the graph,
which knows about entities. Everything in Sections~\ref{sec:idw}
and~\ref{sec:incremental} is an operation on chunk pairs pushed through $\mu$.

\subsection{The additive seam}

Extraction writes into the graph through an adapter method that ingests a list
of nodes and a list of edges and merges them idempotently, keyed on a
deterministic identifier. The similarity pass uses \emph{the same} method, with
an empty node list: every endpoint of an \EMB{} edge is by construction a node
the graph already has, ingested by this run's extraction or an earlier one.

This is the structural reason the derived layer cannot corrupt the extracted
one. The pass has no delete path, no property-overwrite path and no node-creation
path; the only mutation available to it is the addition of edges carrying a
label that extraction never emits. A failure anywhere in the pass therefore
leaves a graph that is smaller than intended but never wrong.

\subsection{The seven stages}

Figure~\ref{fig:stages} shows where the similarity pass sits and what it does.
The remaining sections take the stages in order: \emph{embed} and \emph{persist}
in Section~\ref{sec:embedding}; \emph{$k$-NN}, \emph{expand} and
\emph{accumulate} in Section~\ref{sec:idw}; \emph{select} and \emph{link} in
Sections~\ref{sec:incremental} and~\ref{sec:impl}.

Algorithm~\ref{alg:master} is the same thing written out end to end, in words,
for the arrival of one document. Every misreading of the method we have met comes
either from the granularity --- what is compared against what --- or from
stopping where candidate pairs exist, three stages before any edge is written.

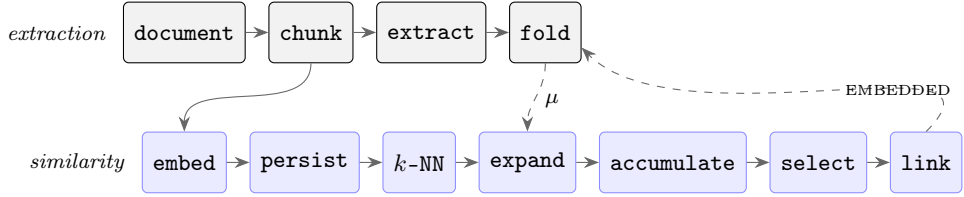
\begin{figure}[t]
\centering
\begin{tikzpicture}[
  node distance=3mm,
  box/.style={draw, rounded corners=2pt, minimum height=8mm, align=center,
              font=\footnotesize\ttfamily, inner xsep=4pt},
  ext/.style={box, fill=black!5},
  sim/.style={box, fill=blue!8, draw=blue!45},
  ar/.style={-{Stealth[length=2mm]}, draw=black!60}
]
\node[ext] (doc)   {document};
\node[ext, right=of doc]  (chunk) {chunk};
\node[ext, right=of chunk](extr)  {extract};
\node[ext, right=of extr] (fold)  {fold};

\node[sim, below=9mm of doc]    (embed)  {embed};
\node[sim, right=of embed]      (pers)   {persist};
\node[sim, right=of pers]       (knn)    {$k$-NN};
\node[sim, right=of knn]        (exp)    {expand};
\node[sim, right=of exp]        (acc)    {accumulate};
\node[sim, right=of acc]        (sel)    {select};
\node[sim, right=of sel]        (link)   {link};

\draw[ar] (doc) -- (chunk);
\draw[ar] (chunk) -- (extr);
\draw[ar] (extr) -- (fold);
\draw[ar] (embed) -- (pers);
\draw[ar] (pers) -- (knn);
\draw[ar] (knn) -- (exp);
\draw[ar] (exp) -- (acc);
\draw[ar] (acc) -- (sel);
\draw[ar] (sel) -- (link);

\draw[ar] (chunk.south) to[out=-90,in=90] (embed.north);
\draw[ar, dashed] (fold.south) to[out=-90,in=90] node[right, font=\scriptsize,
      pos=0.55] {$\mu$} (exp.north);
\draw[ar, dashed] (link.north) to[out=45,in=-45]
      node[right, font=\scriptsize, pos=0.5] {\ \EMB{}} ($(fold.east)+(0,-2mm)$);

\node[font=\scriptsize\itshape, left=1mm of doc]   {extraction};
\node[font=\scriptsize\itshape, left=1mm of embed] {similarity};
\end{tikzpicture}
\caption{The similarity pass (blue) runs after extraction and writes back
through the same additive seam. It consumes the chunk sequence and the
membership map $\mu$, and emits one edge label. It has no path by which it can
alter an extracted node or edge.}
\label{fig:stages}
\end{figure}

\begin{algorithm}[p]
\caption{What happens when one document arrives. Steps 1--5 are extraction and
are conventional; steps 6--14 are the similarity pass. Three lines carry most of
the design. Step~8 is an \emph{exhaustive} comparison against every stored
paragraph, not a search through them. Step~11 is the inverse-distance weighted
average --- the scoring rule of Section~\ref{sec:idw} --- and it is what turns a
set of candidate pairs into a number. Step~12 saves that number \emph{un-gated},
which is why step~13 can promote a pair that this document never mentions.}
\label{alg:master}
\small
\setlength{\tabcolsep}{4pt}
\begin{tabular}{@{}r@{\ }p{0.655\textwidth}@{\hspace{2em}}l@{}}
\toprule
\multicolumn{3}{@{}p{0.96\textwidth}@{}}{%
\textbf{Input:} one new document.\quad
\textbf{Existing state:} the graph, the canonical ontology, the table of
paragraph vectors, and one running total per entity pair.} \\
\midrule
\multicolumn{3}{@{}l}{\itshape Extraction} \\
 1 & Split the document into paragraphs. & {\color{codecm}\code{chunk}} \\
 2 & Ask the language model, once per paragraph, which entities and which
     relations \emph{that paragraph} states. & {\color{codecm}\code{extract}} \\
 3 & Fold the proposed type names into the canonical ontology, so that
     \code{Company} and \code{Organization} become one type
     (Section~\ref{subsec:folding}). & {\color{codecm}\code{fold}} \\
 4 & Record which entities were named in which \emph{paragraph}. This map is the
     only bridge between the text and the graph
     (Section~\ref{subsec:membership}). & {\color{codecm}\code{$\mu$}} \\
 5 & Write the entities and relations into the graph. Everything above has now
     landed; nothing below can alter it (Section~\ref{sec:pipeline}). & {\color{codecm}\code{ingest}} \\
\addlinespace[2pt]
\midrule
\multicolumn{3}{@{}l}{\itshape Similarity pass} \\
 6 & Turn each \emph{paragraph} --- not the document --- into one vector of fixed
     width, $240$ or $768$ numbers, scaled to unit length
     (Section~\ref{sec:embedding}). & {\color{codecm}\code{embed}} \\
 7 & Store those vectors alongside every paragraph vector already kept for this
     graph. & {\color{codecm}\code{persist}} \\
 8 & For each new paragraph, measure its similarity to \emph{every} paragraph
     vector in the store, and keep the $20$ most similar that clear the gate
     $\cos \geq 1-\beta$. This is an exhaustive comparison, not a search: there
     is no starting point and no walking from one neighbour to the next
     (Remark~\ref{rem:notasearch}). & {\color{codecm}\code{$k$-NN}} \\
 9 & Also pair each new paragraph with \emph{itself}, at perfect similarity. Two
     entities named in the same paragraph are the strongest evidence there is,
     and this is how that evidence enters (Section~\ref{sec:idw}). & {\color{codecm}\code{self-pair}} \\
10 & Turn every paragraph pair into entity pairs: one entity from each side,
     counted at most once per paragraph pair
     (Algorithm~\ref{alg:accumulate}). & {\color{codecm}\code{expand}} \\
11 & Score each entity pair by a \textbf{weighted average, in inverse-distance
     weighting fashion}. Each paragraph pair supporting it contributes a
     \emph{value} that falls as the two paragraphs grow apart, and contributes it
     with a \emph{weight} of one over that same distance --- so near evidence
     dominates far evidence. The pair's score is the weighted mean of those
     values, $\sumwv/\sumw$ of \eqref{eq:weight}. & {\color{codecm}\code{IDW}} \\
12 & Add this document's contribution into the running total held for each
     entity pair, and save the totals \textbf{without filtering any of them}
     (Section~\ref{sec:incremental}). & {\color{codecm}\code{accumulate}} \\
13 & Re-read the totals of every pair touched --- their whole history across the
     corpus, not just this document --- and keep a pair if its weighted average
     is at least $\theta$ \emph{and} it is in the top $n$ of either endpoint. & {\color{codecm}\code{select}} \\
14 & Write the survivors as \EMB{} edges carrying their weight and their support
     count, through the same additive seam extraction used. & {\color{codecm}\code{link}} \\
\bottomrule
\end{tabular}
\end{algorithm}

\begin{remark}[The neighbour step is a comparison, not a search]
\label{rem:notasearch}
Readers arriving from the vector-database literature reliably import a mental
model this pass does not use. The difference is worth stating plainly, because
all three parts of it change what the method guarantees.

\emph{The unit is the paragraph, not the document.} A document is never embedded,
never has neighbours, and never appears in the membership map except as half of a
key. Step~8 finds the twenty most similar \emph{paragraphs} to each arriving
paragraph, and those may come from twenty different documents or from one.

\emph{The search is exhaustive by default.} The tempting alternative --- pick an
entry point, move to whichever neighbour is closer to the query, repeat until no
neighbour improves --- is the greedy descent of the navigable-small-world family
\citep{malkov2020}, and it is what an approximate index does. Below roughly
$10^5$ stored vectors it is not what runs here, and it would not pay if it did.
Step~8 is a single self-join over the whole vector table with a top-$k$
aggregate: every stored vector is scored, and the twenty returned are the true
twenty. The reason this is affordable is asymmetry, not cleverness --- only the
arriving document's paragraphs are on the query side, typically tens of vectors,
so the cost is linear in corpus size with a very small constant
(Section~\ref{subsec:maxby}).

Past that threshold an HNSW index takes over the \emph{candidate} half of the
step, and the guarantee weakens in exactly one place. The twenty returned may no
longer be the true twenty; every pair that is returned is still scored by the
same exact inner product, so no weight is ever an estimate. What that costs was
measured on the written edges rather than on the retrieved neighbours, and the
distinction is not cosmetic: the neighbours an approximate search drops sit in
the far tail of the top-$k$, which the threshold and the per-endpoint cap of
step~12 discard in any case. Section~\ref{subsec:cost} reports
$99.6$--$100\%$ agreement on the edge set where chunk-level recall reads
$98.7$--$99.8\%$. Section~\ref{sec:limitations} states what remains.

\emph{Candidate pairs are not edges.} It is natural to read steps 8--10 as the
method and to expect an edge for each entity pair spanning a neighbouring
paragraph. Three stages separate the two. The pair must be \emph{scored} by the
weighted average of step~11; the score is \emph{accumulated} into a total that
spans every document ever ingested, not just this one; and only then is it
\emph{selected} against a threshold and a per-endpoint cap. A pair may therefore
appear in many documents and never earn an edge, or earn one long after the
document that first suggested it.
\end{remark}

\begin{remark}[A degenerate input, excluded upstream]
The pipeline also supports a whole-document extraction mode in which the entire
document is a single chunk. That mode is incompatible with the similarity pass
in a specific and total way: with one chunk, every entity in the document
co-occurs with every other at $\cos = 1$, and the pass emits an all-pairs clique
rather than a similarity signal. The system resolves this by \emph{forcing}
chunked extraction whenever the similarity pass is requested, and reporting that
it did so, rather than by having the pass refuse to run. The reasoning is that
the more specific instruction should win: a user who asks for similarity edges
should receive them, not an explanation of why a setting on another panel
prevented it.
\end{remark}

%% file: sections/04-embedding.tex
\section{Embedding: two widths, pinned per graph}
\label{sec:embedding}

\subsection{Normalisation is a precondition, not a preference}

Chunk vectors are $L_2$-normalised at write time. This is not a stylistic choice
about how to store floating-point data; two later results depend on it.

The first is computational. On unit vectors the inner product \emph{is} the
cosine, so the $k$-NN self-join of Section~\ref{sec:impl} can use the store's
raw \code{array\_inner\_product} primitive with no normalising divisions in the
inner loop --- and that inner loop runs over every chunk in the graph for every
chunk of an arriving document.

The second is analytical. For unit vectors $a$ and $b$,
\begin{equation}
  \lVert a - b \rVert^2 \;=\; \lVert a \rVert^2 + \lVert b \rVert^2 - 2\,a\cdot b
  \;=\; 2\,(1 - \cos),
  \label{eq:chord}
\end{equation}
which is what turns the cosine-based weight of Section~\ref{sec:idw} into a
recognisable Shepard weight and gives the value term a geometric reading. Without
normalisation neither identity holds and the IDW formulation loses its
justification.

\subsection{Matryoshka widths, and why exactly two}

The embedder is \code{gemini-embedding-001} \citep{lee2025}, a Matryoshka
model~\citep{kusupati2022}. The property that matters is that requesting an
output of $N$ dimensions returns exactly the first $N$ components of the full
vector: width is a request-time parameter, not a retraining decision, and a
narrower vector is a prefix of a wider one rather than a different embedding.

The system exposes two widths, $240$ and $768$, and nothing between them. Two
constraints produce that answer from opposite directions.

\paragraph{Storage forces discreteness.} The analytical store types a vector
column as a fixed-width array, so each supported width requires its own physical
table (\code{xgraph\_chunk\_vectors\_240} and \code{xgraph\_chunk\_vectors\_768}).
Offering a continuous choice would mean a table per width in use. The set of
widths is therefore small by construction, and the only question is which ones.

\paragraph{Measurement picks the members.} Table~\ref{tab:dims} gives the
agreement study: for each candidate width, the fraction of emitted \EMB{} edges
that match those produced by the full $3072$-dimensional reference on the same
corpus and the same parameters. At $768$ the layer is substantially the same
layer --- $92\%$ of edges agree --- at a quarter of the storage. At $240$ it is
recognisably the same layer, $72\%$ agreeing, at under a third of $768$'s
storage and under a thirteenth of the reference's. $240$ is the shipped default
and $768$ the fidelity option, because in the operational setting the pass is
used in --- a graph that keeps growing, where every chunk ever ingested stays
resident and is scanned by every subsequent document --- the vector table is the
part that grows without bound, and a $28\%$ disagreement on a layer that is
explicitly advisory costs less than a $3.2\times$ store.

\begin{table}[t]
\centering
\caption{Agreement of the emitted \EMB{} edge set with the full
$3072$-dimensional reference, and the storage each width implies per chunk at
$4$-byte floats. The two shipped widths are the ones that bracket the useful
trade-off; nothing between them is offered because each width needs its own
fixed-width physical table.}
\label{tab:dims}
\begin{tabular}{lrrrl}
\toprule
Width & Edge agreement & Bytes/chunk & Relative storage & Status \\
\midrule
$3072$ & $100\%$ (reference) & $12{,}288$ & $1.00\times$  & not offered \\
$768$  & $92\%$              & $3{,}072$  & $0.250\times$ & fidelity option \\
$240$  & $72\%$              & $960$      & $0.078\times$ & default \\
\bottomrule
\end{tabular}
\end{table}

\begin{remark}
Table~\ref{tab:dims} measures \emph{agreement between configurations}, not
accuracy. A width that disagrees with the reference on $28\%$ of edges is not
thereby wrong on $28\%$ of edges; it is differently right or differently wrong,
and which of those it is cannot be determined without a labelled set that does
not exist here. What the table does support is the operational claim it was
taken for: at $768$ the compression is nearly free, and at $240$ it is a
deliberate, quantified trade rather than a guess.
\end{remark}

\subsection{Pinning width and model per graph}

Both the width and the embedder identity are recorded against the graph on first
use, and a later request that conflicts with the recorded value is
\emph{refused} rather than silently coerced.

The reason is that a graph's vectors share one table and one geometry. Mixing
widths in a fixed-width table is impossible; mixing embedders in one table is
worse, because it is possible. Two models' vectors of the same length will
happily produce an inner product, and that number is meaningless --- there is no
shared space in which the angle between them means anything. A cosine that is
merely meaningless rather than erroneous is exactly the kind of defect that
never announces itself, so the resolution functions for width and for model both
raise on conflict.

The rest of the surface is forgiving by contrast. Every numeric option --- $k$,
$\beta$, the write threshold, the per-endpoint cap --- is \emph{clamped} to its
valid range rather than rejected, on the principle that a malformed number in an
options panel should not fail an extraction that otherwise succeeded. An
unrecognised embedder identity likewise falls back to the graph's pinned value,
since an unknown identifier is usually a stale browser tab rather than a
decision. Only a genuine conflict raises.

\subsection{Persisting a vector, and what is kept beside it}

Each chunk is stored as \code{(graph, doc\_uri, chunk\_index, vec, char\_len,
preview)}, upserted so that re-processing a document is idempotent, together
with the membership rows of Section~\ref{subsec:membership}.

The \code{preview} field --- a whitespace-collapsed $240$-character head of the
paragraph --- exists so that the stored layer can be inspected as paragraphs
rather than as anonymous vectors, and its size is the whole point: long enough to
recognise a paragraph, short enough that the store remains a vector store and
does not become a second copy of the corpus. Vectors may also be discarded after
the pass completes if the caller does not intend to extend the graph further;
the accumulators of Section~\ref{sec:incremental} survive independently, so
discarding vectors costs the ability to add \emph{new} documents cheaply, not
the weights already derived.

%% file: sections/05-idw.tex
\section{Scoring: inverse-distance weighting over cosine neighbourhoods}
\label{sec:idw}

This section is the core of the article. We set out the formulation, derive it
on unit vectors, prove the constraint it places on the write threshold, and
document two plausible variants that do not work --- one of which shipped first
and had to be replaced.

\subsection{Definitions}

\begin{definition}[Chunk pair]
A \emph{chunk pair} is a triple $(c_i, c_j, s)$ where $c_i$ and $c_j$ are chunk
keys and $s = \cos(x_i, x_j)$ is the cosine of their embeddings. Two sources
produce chunk pairs: the top-$k$ neighbour query, and same-chunk co-occurrence,
which enters as the self-pair $(c, c, 1)$.
\end{definition}

\begin{definition}[Node pair evidence]
Let $\mu$ be the membership map of Section~\ref{subsec:membership}. A chunk pair
$(c_i, c_j, s)$ is \emph{evidence for} the unordered node pair $\{a,b\}$ when
$a \in \mu(c_i)$ and $b \in \mu(c_j)$ (or symmetrically), with $a \neq b$. Node
pairs are canonicalised by lexicographic order, so $\{a,b\}$ is stored once.
\end{definition}

\subsection{The formulation}

For a chunk pair at cosine $s$, write the distance, weight and value as
\begin{equation}
  d \;=\; 1 - s,
  \qquad
  w \;=\; \frac{1}{\varepsilon + d},
  \qquad
  v \;=\; 1 - \sqrt{\tfrac{d}{2}},
  \label{eq:delta}
\end{equation}
with $\varepsilon = 10^{-6}$ and $s$ clamped to $[-1, 1]$. The weight of a node
pair is the inverse-distance weighted mean of the value over all its evidence:
\begin{equation}
  \mathrm{weight}(\{a,b\})
  \;=\;
  \frac{\sum_{e \in E(a,b)} w_e\, v_e}{\sum_{e \in E(a,b)} w_e}
  \;=\;
  \frac{\sumwv}{\sumw},
  \label{eq:weight}
\end{equation}
where $E(a,b)$ is the multiset of chunk pairs that are evidence for $\{a,b\}$.
The implementation stores the triple $(\sumwv, \sumw, \textit{support})$ per node
pair, never the individual terms; Section~\ref{sec:incremental} shows why that
representation is the whole incremental story.

\subsection{Why this is Shepard's method exactly}

\begin{proposition}[Shepard $p = 2$]
On $L_2$-normalised vectors, the weight $w = 1/(\varepsilon + d)$ of
\eqref{eq:delta} is Shepard inverse-squared-distance weighting, and the constant
relating them cancels in \eqref{eq:weight}.
\end{proposition}

\begin{proof}
By \eqref{eq:chord}, $\lVert a-b\rVert^2 = 2(1-\cos) = 2d$, so
$d = \lVert a-b\rVert^2/2$ and
\[
  w \;=\; \frac{1}{\varepsilon + \tfrac{1}{2}\lVert a-b\rVert^2}
    \;=\; \frac{2}{2\varepsilon + \lVert a-b\rVert^2},
\]
which is Shepard's $\lVert a-b\rVert^{-p}$ with $p=2$, regularised at the
origin, times the constant $2$. Since \eqref{eq:weight} is a ratio in which
every weight carries that same constant, it cancels.
\end{proof}

The value term has an equally direct reading. Substituting $d = \lVert
a-b\rVert^2/2$ into \eqref{eq:delta},
\begin{equation}
  v \;=\; 1 - \sqrt{\frac{\lVert a-b\rVert^2}{4}}
    \;=\; 1 - \frac{\lVert a-b \rVert}{2},
  \label{eq:vchord}
\end{equation}
the chord distance rescaled to $[0,1]$: unit vectors are at most $2$ apart, so
$v = 1$ for identical direction, $v = \tfrac12$ for orthogonal, $v = 0$ for
opposed. Unlike the alternative considered next, this holds \emph{without}
reference to any gate.

\subsection{Why the affine value term fails behind a gate}

The neighbourhood is gated: the $k$-NN query admits only neighbours with
$\cos \geq 1 - \beta$, with $\beta = 0.35$ by default. The natural first choice
for a value term --- and the one this design inherited, from the ad-hoc node
vector embedding of \citet{karamete2025embedding} --- is the affine rescaling
$v_{\mathrm{aff}} = (1+\cos)/2 = 1 - d/2$. It maps $[-1,1]$ to $[0,1]$ and is
monotone, which appears sufficient.

It is not, and the gate is why.

\begin{proposition}[The affine term collapses onto the gate]
\label{prop:affine}
Under the gate $d \leq \beta$, the affine value satisfies
$v_{\mathrm{aff}} \in [\,1 - \beta/2,\; 1\,]$, an interval of width $\beta/2$.
At $\beta = 0.35$ this is $[0.825, 1]$, of width $0.175$.
\end{proposition}

\begin{proof}
Immediate from $v_{\mathrm{aff}} = 1 - d/2$ and $0 \leq d \leq \beta$.
\end{proof}

Since \eqref{eq:weight} is a convex combination of the $v_e$, every edge weight
inherits that interval. The consequence is not a loss of precision but a loss of
the threshold: a gate at $\beta = 0.35$ makes every weight at least $0.825$, so a
write threshold set anywhere at or below $0.825$ admits everything.

The rescaled chord distance \eqref{eq:vchord} does not collapse, because
$\sqrt{\cdot}$ is steep near zero exactly where the gate concentrates the mass.
Table~\ref{tab:valueterms} shows the two forms on the two cases that motivated
the change.

\begin{table}[t]
\centering
\caption{The two value terms on the gated range. The affine form separates the
strong and the marginal neighbourhood by $0.100$; the chord form separates them
by $0.163$ and, more importantly, places both far above a threshold that can
actually be set below them.}
\label{tab:valueterms}
\begin{tabular}{lrrrr}
\toprule
Case & $\cos$ & $d$ & $v_{\mathrm{aff}} = 1-\tfrac{d}{2}$ & $v = 1-\sqrt{d/2}$ \\
\midrule
same chunk                  & $1.00$ & $0.00$ & $1.000$ & $1.000$ \\
$50$ neighbours, strong     & $0.90$ & $0.10$ & $0.950$ & $0.776$ \\
$200$ neighbours, marginal  & $0.70$ & $0.30$ & $0.850$ & $0.613$ \\
gate boundary ($\beta=0.35$)& $0.65$ & $0.35$ & $0.825$ & $0.582$ \\
\bottomrule
\end{tabular}
\end{table}

\subsection{The admissible threshold}

Proposition~\ref{prop:affine} generalises into the constraint that governs how
the write threshold may be configured.

\begin{proposition}[Threshold floor]
\label{prop:floor}
With the value term \eqref{eq:delta} and the gate $d \leq \beta$, every node
pair weight satisfies
\[
  \mathrm{weight} \;\geq\; 1 - \sqrt{\beta/2}.
\]
Consequently a write threshold $\theta \leq 1 - \sqrt{\beta/2}$ filters nothing.
\end{proposition}

\begin{proof}
Every term of \eqref{eq:weight} has $w_e > 0$, so the ratio is a convex
combination of the $v_e$ and is bounded below by $\min_e v_e$. Each $v_e = 1 -
\sqrt{d_e/2}$ is decreasing in $d_e$, and the gate gives $d_e \leq \beta$;
therefore $v_e \geq 1 - \sqrt{\beta/2}$ for every $e$, and so does the
combination.
\end{proof}

At the default $\beta = 0.35$ the floor is $1 - \sqrt{0.175} = 0.582$. The
default threshold is $0.7$, chosen to sit above it. This is the practical
content of Proposition~\ref{prop:floor}: the threshold and the gate are not
independent knobs, and a user interface that presents them as independent must
at least display the floor that the current $\beta$ implies. The system does,
and the floor it displays is $1-\sqrt{\beta/2}$; an earlier version displayed
$1 - \beta/2$, which is the floor for the value term that was replaced.

\subsection{The near-miss: \texorpdfstring{$L_2$}{L2} in the weight}

There is a second, more tempting repair for the compression of
Proposition~\ref{prop:affine}: leave the value term affine and put the square
root in the \emph{weight} instead, $w' = 1/(\varepsilon + \lVert a-b\rVert)$.
This is Shepard $p=1$, and it is the wrong fix for two independent reasons.

\paragraph{It flattens the weighting it was meant to sharpen.} Compare a very
close neighbour at $d = 10^{-3}$ against one at the gate boundary $d = \beta =
0.35$. Under $w = 1/(\varepsilon+d)$ the ratio of their influence is
$1000 : 2.86$, a spread of $350\times$. Under $w' = 1/(\varepsilon+\sqrt{2d})$
it is $22.4 : 1.20$, a spread of $19\times$. Inverse-distance weighting exists to
let near evidence dominate far evidence; reducing the dynamic range by more than
an order of magnitude defeats the mechanism.

\paragraph{It does not address the same-chunk term either.} The stated
motivation for moving the $L_2$ term is usually the $\cos = 1$ spike. But at
$\cos = 1$ we have $\lVert a-b\rVert = 0$ as well as $d = 0$, so $w' =
1/\varepsilon = 10^6$ exactly as before. The spike is a property of
$\varepsilon$, not of the exponent, and changing the exponent leaves it
untouched.

\subsection{The same-chunk spike is intentional}

Two entities extracted from the same paragraph enter as a self-pair at $\cos =
1$, giving $w = 1/\varepsilon = 10^6$ and $v = 1$. Such a term dominates every
near-neighbour term in \eqref{eq:weight} by three to six orders of magnitude, so
a node pair with even one shared paragraph will score close to $1$.

This is the intended behaviour, and it encodes the ordering we actually believe:
sharing a paragraph is much stronger evidence of a relationship than being
mentioned in two paragraphs that resemble each other. The mechanism is worth
naming explicitly, because it means $\varepsilon$ is not a numerical guard
against division by zero --- it is the parameter that sets how much a shared
paragraph outweighs a similar one, and $10^{-6}$ is a decision.

\subsection{Expansion and accumulation}

Algorithm~\ref{alg:accumulate} gives the expansion of chunk pairs into node-pair
accumulators --- steps 10 to 12 of Algorithm~\ref{alg:master} at the level of the
inner loop. Two guards in it are easy to omit and each produces a systematic
overcount.

\begin{algorithm}
\caption{Accumulate --- chunk pairs into node-pair IDW accumulators}
\label{alg:accumulate}
\begin{tabbing}
xxx\=xxx\=xxx\=xxx\=\kill
\textbf{input:} chunk pairs $P$, membership $\mu$ \\
\textbf{output:} $A : \{a,b\} \mapsto (\sumwv, \sumw, \textit{support})$ \\[2pt]
$A \gets \emptyset$ \\
\textbf{for} $(c_i, c_j, s) \in P$ \textbf{do} \\
\> $L \gets \mu(c_i)$; \quad $R \gets \mu(c_j)$ \\
\> \textbf{if} $L = \emptyset$ \textbf{or} $R = \emptyset$ \textbf{then continue} \\
\> $(w, wv) \gets \big(\tfrac{1}{\varepsilon+1-s},\ \tfrac{1-\sqrt{(1-s)/2}}{\varepsilon+1-s}\big)$ \\
\> $\textit{seen} \gets \emptyset$ \qquad \textit{// reset per chunk pair} \\
\> \textbf{for} $a \in L$, $b \in R$ \textbf{do} \\
\> \> $q \gets \textsc{canonical}(a,b)$ \\
\> \> \textbf{if} $q = \bot$ \textbf{or} $q \in \textit{seen}$ \textbf{then continue} \\
\> \> $\textit{seen} \gets \textit{seen} \cup \{q\}$ \\
\> \> $A[q] \mathrel{+}= (wv,\ w,\ 1)$ \\
\textbf{return} $A$
\end{tabbing}
\end{algorithm}

\paragraph{Guard 1: one contribution per chunk pair.} The \textit{seen} set is
reset for each chunk pair and prevents a single chunk pair from contributing to
the same node pair twice. It is required by the self-pair case: when $c_i = c_j$,
the double loop enumerates both $(a,b)$ and $(b,a)$, which canonicalise to the
same node pair, so every same-chunk co-occurrence would otherwise be counted
twice --- and since these are the dominant terms, the resulting weight would be
wrong for exactly the pairs the layer cares most about. The same applies to a
symmetric neighbour pair whose two chunks share entities.

\paragraph{Guard 2: mirrored chunk pairs from the query side.} The $k$-NN query
restricts only its \emph{query} side to the arriving document. Two chunks of that
same document that are mutual nearest neighbours therefore return as both
$(c_0, c_1, s)$ and $(c_1, c_0, s)$. Because \textit{seen} is per chunk pair, it
does not deduplicate across them, and the node pair is counted twice. The fix is
a deduplication of chunk pairs on the unordered key before accumulation. Note
that this affects intra-document pairs only: a cross-document neighbour is
queried from one side only and has no mirror. Cosine is symmetric, so discarding
the duplicate loses nothing --- both directions carry the same value.

\subsection{Selection}

Two gates are applied to what is \emph{written}, and neither is applied to what
is accumulated.

\paragraph{Threshold.} A candidate survives if $\mathrm{weight} \geq \theta$,
with $\theta = 0.7$ by default, admissible under Proposition~\ref{prop:floor}.

\paragraph{Per-endpoint cap, with union semantics.} A candidate survives if it is
among the top $n$ by weight of \emph{either} endpoint --- a union, not an
intersection. The asymmetry is deliberate: under intersection semantics, an
entity whose only similarity edge happens to be some hub's fortieth-best would
lose that edge and be orphaned, so the cap would systematically strip the
periphery to tidy up the centre. Under union semantics a quiet node keeps its
one edge while a popular node is still prevented from radiating unboundedly.

Finally, each surviving pair becomes an \EMB{} edge carrying its weight, its
support count and a kind marker, with a deterministic identifier derived from
the endpoint pair and the label so that re-running the pass merges rather than
duplicates.

%% file: sections/06-incremental.tex
\section{A layer that evolves with the corpus}
\label{sec:incremental}

The pass is designed for a graph that is never finished. Documents arrive one at
a time, in no particular order, and the ontology they fold into gains types,
aliases and facets as they do. This section states the properties that make the
similarity layer well behaved under those conditions, and it argues that they
follow from a single representational choice rather than from any machinery.

\subsection{The accumulator is a commutative monoid}

Each node pair carries the triple
\[
  \alpha \;=\; (\sumwv,\ \sumw,\ \textit{support}) \;\in\;
  \mathbb{R}_{\geq 0} \times \mathbb{R}_{\geq 0} \times \mathbb{N},
\]
combined componentwise by addition, with identity $(0,0,0)$. Componentwise
addition on this set is associative and commutative; $(\mathcal{A}, +, 0)$ is
therefore a commutative monoid.

The edge weight is not stored. It is the function
$\pi(\alpha) = \sumwv / \sumw$ applied at read time, with $\pi(0,0,0) = 0$ by
convention. Keeping the accumulator rather than the weight is what makes the
next two results available; storing a weight and trying to update it in place
would require either the individual terms or a running count, and would make the
merge dependent on how the batches were cut.

\begin{proposition}[Order independence]
\label{prop:order}
Let $D_1, \ldots, D_m$ be documents and let $\mathrm{acc}(D)$ denote the
accumulator map that Algorithm~\ref{alg:accumulate} produces for $D$ against a
fixed corpus. Then the stored state after processing the documents in any
permutation is the same, and equals $\sum_{t} \mathrm{acc}(D_t)$ pairwise.
Consequently the weight of every node pair is independent of arrival order.
\end{proposition}

\begin{proof}
The stored state is the pairwise sum of the per-document accumulator maps in the
monoid $\mathcal{A}$, and addition there is associative and commutative, so any
permutation yields the same sum. $\pi$ is a function of that sum alone.
\end{proof}

Proposition~\ref{prop:order} holds only for the accumulator, not for the emitted
edge set at intermediate times: which edges exist \emph{after the third of five
documents} does depend on order, because the threshold has been applied to a
partial sum. That is a statement about a snapshot, not about the final state, and
it is exactly the behaviour one wants --- the layer reflects the evidence seen so
far.

The phrase \emph{against a fixed corpus} in Proposition~\ref{prop:order} is
load-bearing, and it is worth separating the two halves it distinguishes, because
only one of them is order-free. The \emph{sum} is: adding the same terms in a
different order gives the same accumulator. Which terms get \emph{proposed} is
not, because a document's top-$k$ query runs against the vector table as it
stands at that moment, and an early document sees a smaller table than a late
one. We can measure the split. Recomputing the whole corpus in one shot yields
$1{,}415{,}574$ node pairs against the $1{,}517{,}795$ the incremental run
actually left behind. On the $1{,}051{,}617$ pairs common to both, the weights
agree \emph{to the last bit} for $89.8\%$, with a mean absolute difference of
$8.1 \times 10^{-4}$; mean support is $1.21$ against $1.23$. The residual is
almost entirely candidate generation: $363{,}957$ pairs exist only in the batch
run and $466{,}178$ only in the incremental one, because the neighbourhoods
differed. So the arithmetic is order-independent as claimed, and the pass as a
whole is order-\emph{sensitive} in what it considers --- which is the price of
never recomputing, and is why the incremental run finds \emph{more} than a single
batch pass with the same parameters rather than less.

\begin{remark}[The self-similarity term is excluded, not incorporated]
The $k$-NN query excludes the trivial pair $(c, c)$ from its results, so the only
$\cos = 1$ terms in the sum are the same-chunk co-occurrence self-pairs
deliberately injected in the expand stage. Without that exclusion every chunk
would be its own nearest neighbour and the corresponding term would be injected
twice.
\end{remark}

\subsection{Persist un-gated: the reason a pair can cross later}

The threshold $\theta$ and the per-endpoint cap are applied to the candidate list
at write time. They are \emph{not} applied before persistence. The accumulators
of every touched pair are stored regardless of how weak the pair currently looks.

This is the property the whole design is arranged around. A pair of entities
mentioned in two mildly similar paragraphs accumulates a small $\sumw$ and a
weight that may sit below $\theta$; it produces no edge and is invisible in the
graph. When a later document supplies a third paragraph that is strongly similar
to both, that pair's accumulator gains a large $w$ term and its weight moves ---
and the edge appears, with a support count that records how many pieces of
evidence were involved. Nothing was recomputed, and no record of the earlier weak
evidence had to be reconstructed, because it was never discarded.

The converse behaviour is worth being explicit about: because the accumulator is
monotone in the additive sense, and because the pass has no delete path, the
weight of a pair can move in either direction as evidence accrues --- a burst of
weak but very numerous evidence lowers a mean that was previously set by a single
strong term --- but the \emph{support} count only rises, and no pair ever leaves
the accumulator table. An edge already written whose weight later falls below
$\theta$ is not retracted by this pass.

\subsection{Cost per document}
\label{subsec:cost}

Let $N$ be the number of chunks in the graph and $c$ the number of chunks in the
arriving document, with $c \ll N$. The query side of the $k$-NN is restricted to
the arriving document while the search side is the whole corpus, so the pass
performs $\Theta(c\,N)$ similarity evaluations exactly --- one pass over the
corpus per new chunk --- and not $\Theta(N^2)$. Over a corpus ingested one
document at a time the total is $\Theta(N^2)$ evaluations in the aggregate, but
distributed as a constant amount of work per document rather than a rebuild.

The expansion stage is bounded by $\sum_{(c_i,c_j)} |\mu(c_i)|\cdot|\mu(c_j)|$,
which is small in practice because $|\mu(c)|$ is the number of entities the
extractor found in one paragraph. The accumulate stage is a merge into a keyed
store and, measured on the running system, dominates the wall clock of the whole
pass --- which is why it is reported as its own progress phase rather than being
folded into the similarity phase, and why the upsert formulation of
Section~\ref{sec:impl} mattered enough to measure.

An approximate index on the search side brings the per-document cost to roughly
$\Theta(c \log N)$ at the usual recall trade-off. One is now built, but only past
a threshold, and the reason for the threshold is that below it there is nothing
to buy. Against a store of $N$ vectors the scan is bound by the dot products
themselves rather than by memory bandwidth; measured per document at four chunks,
it costs $6.5$\,ms at $10^4$ stored vectors, $46.3$\,ms at $10^5$ and $67.0$\,ms
at $10^6$. Extraction, meanwhile, runs at a median of $91.1$ seconds per
document. Even at a million chunks the scan is under a tenth of one percent of
what the document already costs, and no index changes that ratio in a way anyone
would notice --- which is why the index is switched on by corpus size rather than
used unconditionally.

What does eventually bind is the \emph{cumulative} quadratic: ingesting a corpus
one document at a time performs $N^2/2$ similarity evaluations in total, which
measures at about four hours at a million chunks and seventeen days at ten
million. That is the horizon
Section~\ref{sec:limitations} refers to, and it is a statement about an aggregate
ingest rather than about any single insertion.

\paragraph{Alternatives considered.} Several ways of avoiding or accelerating the
scan were measured on this corpus. Every one that \emph{partitions} the vectors
failed; the one that does not --- a navigable small-world graph --- is what ended
up adopted.

Clustering the vectors and scoring pairs by co-membership --- $k$-means and its
spherical variant~\citep{dhillon2001}, in the self-organising tradition
of~\citet{kohonen1990} and the family surveyed by~\citet{jain2010} --- changes the
written graph by about one percent. It is unusable without the weighting of
Section~\ref{sec:idw}, and it forfeits incrementality outright: a centroid is a
function of the whole corpus, so refitting it retroactively invalidates work
already done. Using that same clustering as an \emph{index} rather than as a
score does better than chance, but not by enough: an inverted-file partition at
$K = 32$ recovers $38\%$ of the exact neighbours while scanning $4\%$ of the
corpus, and is still at $86\%$ having scanned $32\%$. Space-partitioning trees
fail outright, because distances here are too concentrated for a pruning bound to
exclude anything~\citep{weber1998,beyer1999}.

Changing the \emph{basis} before partitioning does not rescue it. Clustering on
wavelet coefficients or band energies is no better than clustering on the raw
coordinates --- and a control that randomly permutes the coordinates scores
slightly \emph{higher}, which means the transform basis carries no usable
information about neighbourhood here. Nor does repairing the geometry. The
embeddings are severely anisotropic: the mean of all $1{,}010$ unit vectors has
norm $0.897$, and pairwise cosine has median $+0.806$, interquartile range
$0.039$, and a corpus-wide minimum of $+0.630$ --- no two chunks are even mildly
dissimilar. Centring and renormalising widens the angular spread by $3.5\times$
(median $-0.009$, IQR $0.137$), exactly the isotropy the literature prescribes,
and moves inverted-file recall by less than a percentage point at matched scan
fraction; relative contrast~\citep{beyer1999} in fact \emph{falls}, $0.643$ to
$0.488$. Three partitions, three failures, and the common factor is that this
corpus has no partition structure for a basis to expose.

Greedy descent over a navigable small-world graph~\citep{malkov2020,douze2024}
does not partition, and it is the one approach that clears the bar. Judged where
it matters --- on the \emph{edge set} the pass finally writes, after expansion,
weighting, thresholding and the per-endpoint cap --- an HNSW index reproduces
$99.6\%$ of the exact result at its cheapest search setting and $100\%$ from
$\mathrm{ef}=80$ upward, losing no edge at all; the residual disagreement is nine
extra edges out of $57{,}475$. Chunk-level recall understates this substantially
($98.7\%$ at the same setting), because the neighbours an approximate search
misses are in the far tail of the top-$k$, which the threshold and the cap
discard anyway. Insertion into the index is incremental and lossless, so the
accumulator property of Section~\ref{sec:idw} survives: adding one document costs
$0.003$\,s against $0.185$\,s to rebuild.

What decides whether the index is worth using is not the search at all. Binding a
$240$-component query vector as a statement parameter costs $7.9$\,ms, against
$0.06$\,ms for a trivial round trip, and the engine will not batch an index
probe; inlining the vector as a literal instead takes a query from $11.9$ to
$5.3$\,ms and is what makes the indexed route win. Per document, the exact scan
costs $3.2$/$30.4$/$42.2$\,ms at $10^4$/$10^5$/$10^6$ stored vectors against the
index's $14.2$/$19.9$/$23.8$: the crossover is near $10^5$, and below it the
exact scan is simply faster. Hence the threshold.

Two things that help less than expected are worth recording, since both are the
obvious first guesses. Batching many documents into one query improves the
constant by only $1.5$--$2.0\times$. Expressing the same computation as a BLAS
matrix multiply runs the arithmetic twenty times faster and gains only
$2$--$4\times$ end to end, because the top-$k$ selection it uncovers is itself
$\Theta(cN)$. Exact top-$k$ over all pairs is irreducibly linear in the corpus
however it is expressed; only searching less changes that, which is what an index
does. The supporting measurements are in the project's engineering notes rather
than here, since none of them changed the design.

\subsection{What ``evolving ontology'' means here, and what it costs}

Three things change as the corpus grows: the instance graph gains nodes and
edges; the ontology gains canonical types, aliases and facets
(Section~\ref{subsec:folding}); and the similarity layer gains \EMB{} edges ---
Figure~\ref{fig:ontology-embedded} shows what the third of those does to the
schema after five documents. The
first two are extraction's business. The third is arranged so that it never has
to be told about the first two: the pass reads the membership map and the vector
table, and writes edges between node identifiers. It holds no schema, caches no
type list, and is unaffected by a type being renamed to its canonical form,
because the fold happens before the identifiers it works with are minted.

There is one cost, and it should be stated plainly. The accumulators encode the
value term: $\sumwv$ is a sum of $w \cdot v$ products computed under a specific
$v$. Changing $v$ --- as happened when the affine form of Section~\ref{sec:idw}
was replaced --- invalidates every stored accumulator, because \eqref{eq:weight}
recomputes from history and the history is now in mixed units. There is no
migration for this; the correct response is to re-extract into a fresh graph, and
mixing is not detectable after the fact. A version marker on the accumulator
table would make it detectable, and is the obvious hardening.

\subsection{A per-run bound, not a per-graph invariant}

One honest caveat about the per-endpoint cap. It is applied over the candidates
touched by \emph{this run} --- the node pairs this document's chunks reached ---
not over the graph's whole accumulator table. A node can therefore finish with
more than $n$ \EMB{} edges after many documents, even though no single run gave
it more than $n$.

This is deliberate. Enforcing a true per-graph invariant would mean ranking every
node against the full accumulator table on every extraction, which is a real and
growing cost for a property nothing downstream depends on. The cap exists to
bound the fan-out any one pass can add, and it does that. Reporting it as a
per-graph guarantee would be the error; widening the read to make it one would be
the more expensive error.

\begin{figure}[t]
\centering
\includegraphics[width=0.94\textwidth]{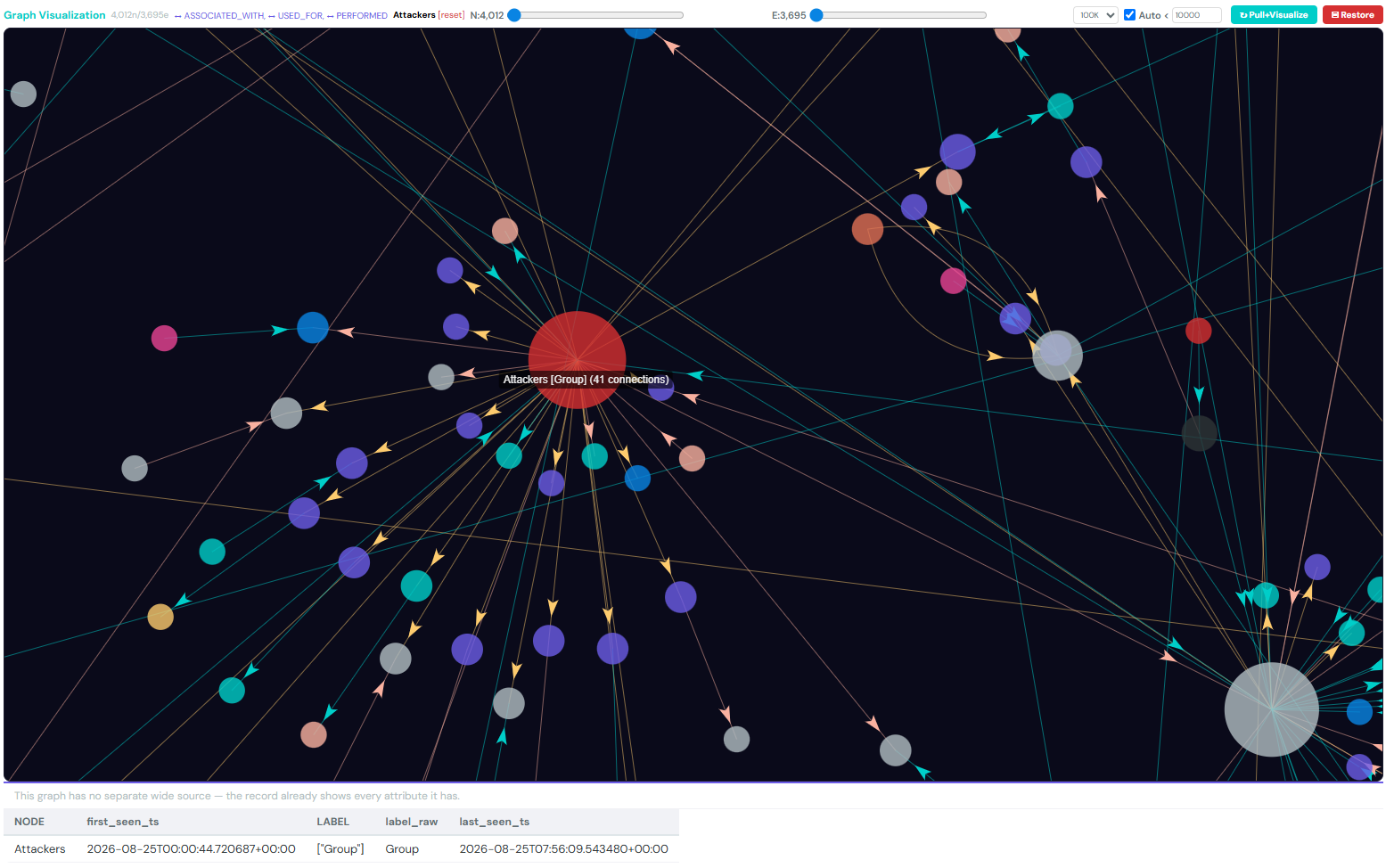}
\caption{The extreme case of the per-run cap, and what a node in one of these
graphs actually holds. \code{Attackers} is the highest-degree node in
\code{threat\_rss} after $246$ documents: $809$ edges, of which $742$ are \EMB{}
and $67$ extracted (\code{PERFORMED} $28$, \code{GAINS\_ACCESS\_TO} $10$,
\code{USED\_FOR} $8$, then a tail). It appears in $37$ chunks across $28$
documents, so a nominal per-run cap of $n = 10$ produced an average of $26$
derived edges per document that mentioned it. That is the union semantics of
Section~\ref{sec:idw} compounding with the per-run scope: the cap keeps each
endpoint's \emph{own} top ten, and a frequently-mentioned node is additionally
carried in on many other endpoints' lists. No single run violated the bound; the
bound simply does not compose. The tooltip's ``$41$
connections'' is its degree within the three edge labels currently selected. The
strip below the canvas is the whole node record, and it is the point of the
skinny-graph design: \code{NODE}, \code{LABEL}, \code{label\_raw} and two
timestamps, with no domain attributes at all --- the workbench says so in as many
words, because for an extracted graph there is no wide relational source to join
back onto. What the documents deposited on this node is not a row; it is a
neighbourhood. The two timestamps are the incremental claim in its most direct
form: \code{Attackers} was minted $44$ seconds into the run and last written
$7$\,h\,$55$\,m later, so almost every accumulator behind those $742$ edges was
updated by a document that arrived after the node already existed.}
\label{fig:hub}
\end{figure}

Figure~\ref{fig:hub} shows the largest instance of this in the corpus run of
Section~\ref{sec:scale}, together with the reason a node can absorb that many
derived edges without the graph carrying any more \emph{data} about it.

%% file: sections/07-implementation.tex
\section{Implementation}
\label{sec:impl}

The pass is implemented in three modules with a strict separation that is worth
describing, because it is what let every claim in Section~\ref{sec:idw} be
pinned by a test that touches no engine, no database file and no network.

\begin{itemize}
\item \textbf{The arithmetic} --- \eqref{eq:delta}, \eqref{eq:weight},
      canonicalisation, Algorithm~\ref{alg:accumulate}, the two selection gates
      and the \EMB{} edge shape --- is a module of pure functions. It has no
      clock, no I/O and no configuration.
\item \textbf{The embedder} resolves the route to the embedding provider, pins
      width and model per graph, and refuses conflicts
      (Section~\ref{sec:embedding}). It also provides a deterministic offline
      embedder used only by tests: it has no semantics whatsoever, and exists so
      that the pass can be exercised end to end reproducibly.
\item \textbf{The orchestration} walks the seven stages. The store and the graph
      adapter are \emph{parameters} rather than module-level lookups, so the
      whole pass runs against fakes.
\end{itemize}

\subsection{\texorpdfstring{Top-$k$}{Top-k}: an aggregate, not a window function}
\label{subsec:maxby}

The neighbour query is a self-join over the vector table, restricted on the
query side to the arriving document and searching the whole graph. Listing~1
gives the generated form.

\begin{lstlisting}[language=SQL,caption={The top-$k$ neighbour self-join. The
top-$k$ is an aggregate over the join, not a ranked window over it.},
captionpos=b]
SELECT q.doc_uri AS q_doc, q.chunk_index AS q_idx,
       max_by({'d': e.doc_uri, 'i': e.chunk_index,
               's': array_inner_product(q.vec, e.vec)},
              array_inner_product(q.vec, e.vec), 20) AS nbrs
  FROM chunk_vectors_768 q, chunk_vectors_768 e
 WHERE q.graph = ? AND q.doc_uri = ?
   AND e.graph = q.graph
   AND NOT (e.doc_uri = q.doc_uri AND e.chunk_index = q.chunk_index)
   AND array_inner_product(q.vec, e.vec) >= 0.65
 GROUP BY q.doc_uri, q.chunk_index
\end{lstlisting}

The natural way to write top-$k$-per-group in modern SQL is a window function:
rank the join output by similarity within each query chunk and keep the first
$k$. Measured at $768$ dimensions over $100$k rows, that form took
$35{,}819$\,ms against this form's $1{,}403$\,ms --- a factor of $25$ --- and
profiling attributed $96\%$ of the difference to the window operator itself
rather than to the dot products. The reason is structural: the window form must
materialise and sort every candidate pair before discarding all but $k$, whereas
the aggregate maintains a bounded top-$k$ heap per group and never materialises
the rest. At these dimensions the join output is the expensive object, so not
building it is worth more than any improvement to the distance computation.

Two implementation notes. The count argument of the aggregate must be a
constant, so $k$ and the similarity floor are interpolated into the statement
text rather than bound as parameters; both are numerically coerced immediately
before interpolation, so nothing unvetted reaches the SQL. And the exclusion of
the trivial self-match is written as a predicate on the pair of keys rather than
on similarity, since two distinct chunks with identical text would otherwise be
excluded as well.

\subsection{Folding the accumulators: one statement, not one per row}

The accumulate stage merges the run's contributions into the stored table with
running sums, \code{sum\_wv = excluded.sum\_wv + sum\_wv} and likewise for the
other two components.

The first implementation issued this as a prepared statement executed once per
row. That runs each row as its own statement and transaction, measured at
$82.3$\,s for a $19{,}900$-pair merge --- a few milliseconds per row, dominating
the entire pass. Rewriting it as a single multi-row \code{VALUES} list makes the
conflict resolution and the insert happen once for the whole batch. The per-pair
arithmetic is unchanged; only the number of statements carrying it differs.

\subsection{Engine neutrality, and declining by name}

The link stage writes through the same adapter method extraction uses, with an
empty node list. Four of the supported engines implement it directly: FalkorDB,
Neo4j and ArangoDB treat a label as a value on the element, so introducing
\EMB{} is an ordinary write, and Kinetica's graph DDL accommodates it as well.

Warehouse-backed property graphs are the exception, and the reason is the model
rather than the product: when a label is part of the schema, there is no additive
element write --- introducing an edge label is a DDL change. The pass detects
this by comparing the adapter's bound method against the base class's, which is
the only reliable test when the base declares the method and raises. It then
returns a report naming the engine and the reason, rather than propagating an
exception the caller would surface as an internal fault. The distinction matters
operationally: ``similarity edges are not supported on this engine'' is
actionable, and a stack trace is not.

The graph's declared direction is read from the store and passed through, rather
than assumed. On engines whose DDL is genuinely directed or undirected, an
\EMB{} write must agree with what extraction already declared for that graph and
must not silently redeclare it.

\subsection{Keeping a cosine artefact out of the answer path}

An \EMB{} edge is a similarity artefact with a numeric weight. It is not a fact,
and the natural-language question-answering path must not traverse it and then
report the result as though it were one.

Suppressing it requires two changes, and the second is the one that is easy to
miss. The obvious one is to drop \EMB{} from the relation-type list handed to the
model. The other is that the schema summary also includes the ontology as
subject--predicate--object triples, and \EMB{} appears there too; dropping it from
the type list alone leaves it fully visible in the triples, and a model shown a
traversable triple will use it. Both are filtered.

This is a policy decision expressed in a prompt boundary rather than in the data,
and it is reversible: the edges remain in the graph, fully visible to direct
Cypher or GQL queries, to the visualiser, and to any consumer that has decided
for itself what a weighted similarity edge means.

\begin{figure}[t]
\centering
\includegraphics[width=\textwidth]{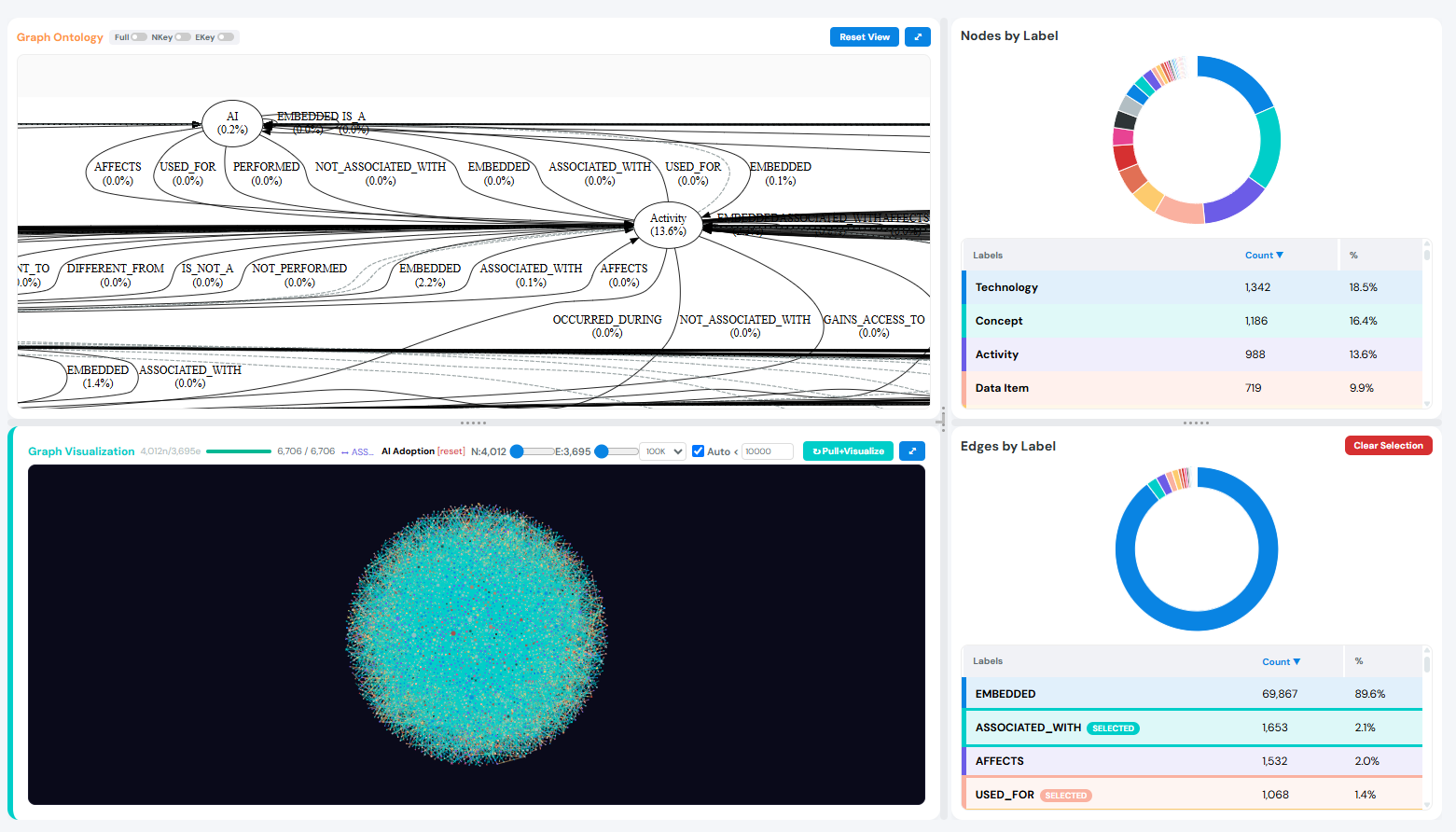}
\caption{The workbench on \code{threat\_rss}, which is the instrument every
measurement in Section~\ref{sec:scale} was read from. \emph{Top left}, the
ontology: the folded canonical types as a graph over the relation types that
join them, each annotated with its share of traffic, so \code{Activity} at
$13.6\%$ is a type a seventh of the graph's nodes carry and the several
\code{EMBEDDED} arcs are the derived layer entering the schema as an ordinary
relation type. It is drawn at a legible zoom deliberately; at $246$ documents the
full folded ontology is too wide to render whole, which is itself the finding of
Section~\ref{subsec:vocab}. \emph{Right}, the two label distributions. The lower
one is the ratio this article keeps returning to, read directly off the graph:
\code{EMBEDDED} $69{,}867$ edges at $89.6\%$, against \code{ASSOCIATED\_WITH}
$1{,}653$ ($2.1\%$), \code{AFFECTS} $1{,}532$ ($2.0\%$) and \code{USED\_FOR}
$1{,}068$ ($1.4\%$). \emph{Bottom left}, the canvas. The two \code{SELECTED}
chips in the edge table are the mechanism behind
Figure~\ref{fig:canvas-extracted}: the derived layer is a layer one turns off,
not a property of the graph one is stuck with.}
\label{fig:workbench}
\end{figure}

Figure~\ref{fig:workbench} shows what that visibility amounts to in practice. The
answer path cannot traverse \EMB{}, but the inspection path is built to make it
the first thing an operator sees, because the ratio between the derived and the
extracted layer is the single number that determines whether a graph built this
way is worth querying.

\subsection{Reproducibility}

Every arithmetic claim in Section~\ref{sec:idw} --- the threshold floor of
Proposition~\ref{prop:floor}, the union semantics of the cap, the per-chunk-pair
deduplication --- is covered by tests over the pure module, with no services
required. The end-to-end pass is exercised against a fake store and a fake
adapter using the deterministic offline embedder, so the seven stages run
identically on a machine with no credentials. Integration tests against live
engines skip rather than fail when the engine is unreachable.

%% file: sections/08-results.tex
\section{Results}
\label{sec:results}

We report three kinds of result: what the pass does to the running example, how
the design decisions of Sections~\ref{sec:embedding}--\ref{sec:impl} measure, and
the default configuration that follows from them. Every number below was taken
against the running system; none is estimated.

\subsection{The running example}

Figures~\ref{fig:extract} and~\ref{fig:connect} are the same three paragraphs
--- one document each, so three chunks --- before and after the pass. Extraction
produces six entities across two components and four stated edges, each
corresponding to a sentence. The two components are $\{$\code{Gulgun},
\code{FSI}$\}$ and $\{$\code{Kaan}, \code{Tan}, \code{BabelStreet},
\code{Bloomberg}$\}$: nothing any of the three paragraphs says connects them.
The similarity pass adds four \EMB{} edges and \emph{no} nodes, joining the
components into one graph.

\begin{table}[t]
\centering
\caption{The running example, before and after. The pass is additive in the
strict sense: node count and stated-edge count are unchanged.}
\label{tab:example}
\begin{tabular}{lrr}
\toprule
 & After extraction & After similarity pass \\
\midrule
Documents                 & $3$  & $3$ \\
Chunks embedded           & $3$  & $3$ \\
Entities (nodes)          & $6$  & $6$ \\
Stated edges              & $4$  & $4$ \\
\EMB{} edges              & $0$  & $4$ \\
Connected components      & $2$  & $1$ \\
\bottomrule
\end{tabular}
\end{table}

The four derived edges carry weights in $[0.72, 0.84]$, all above the default
threshold of $0.7$ and all well above the floor of $0.582$ that
Proposition~\ref{prop:floor} gives for $\beta = 0.35$. The strongest,
\code{Gulgun}\,--\,\code{Kaan} at $0.84$, is the tie the corpus most clearly
implies and no sentence states.

Figures~\ref{fig:extract} and~\ref{fig:connect} are schematics: they show the
tie the pass is for, reduced to the entities a reader can hold at once.
Figure~\ref{fig:run} is what the system actually returns when those same three
paragraphs are fed to it. The extractor is considerably less parsimonious than
the schematic --- $23$ entities and $24$ stated relations across $13$ node
labels, not six entities and four edges --- and the similarity pass adds $86$
\EMB{} edges, taking the graph to $110$. The ratio is the point worth noting:
nearly four in five edges are derived. Three paragraphs that are all about one
family, with every chunk similar to every other, are close to a worst case for
a chunk-similarity signal, and the figure should be read as showing the
mechanism at its least selective rather than as a typical density.

\begin{figure}[!htb]
  \centering
  \includegraphics[width=0.52\textwidth]{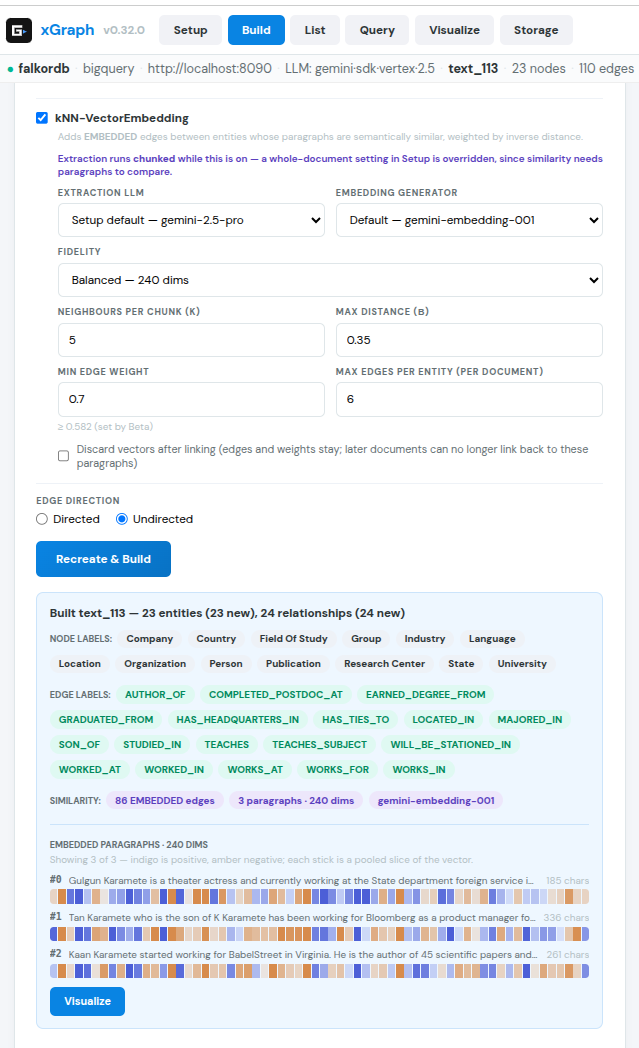}
  \caption{The pass run on the three paragraphs of
  Figures~\ref{fig:extract}--\ref{fig:connect}. The header records the outcome:
  $23$ entities, $24$ relationships, $86$ \EMB{} edges, $110$ edges in total, at
  $240$ dimensions. The hint under the threshold field --- \emph{$\geq 0.582$
  (set by Beta)} --- is Proposition~\ref{prop:floor} evaluated live for the
  configured $\beta$, so an operator cannot set a threshold the gate has already
  made inert. Two settings here are below their defaults: $k = 5$ rather than
  $20$, and a per-endpoint cap of $6$ rather than $10$, which is why the derived
  layer is as dense as it is on so few nodes. The coloured strips are the three
  chunk vectors, mean-pooled to $64$ buckets for display.}
  \label{fig:run}
\end{figure}

\subsubsection*{What the threshold does to it}

Section~\ref{subsec:sweep} sweeps $\theta$ over the full corpus, where the
populations are large enough to be statistical and no individual edge can be
followed. The same question put to these three paragraphs has an answer that can
be checked by hand, and Figure~\ref{fig:example-theta} is it: one graph, drawn
three times, at three thresholds.

The run behind that figure is a fresh one. The graph in Figure~\ref{fig:run} was
an ad-hoc session and no longer exists, and extraction is not deterministic, so
re-running the same three paragraphs gives $19$ entities and $17$ stated
relations rather than $23$ and $24$. Everything else is held at
Figure~\ref{fig:run}'s configuration, including $k = 5$ and the per-endpoint cap
of $6$. The pass adds $124$ \EMB{} edges. Extraction alone leaves \emph{two}
connected components --- the property Figures~\ref{fig:extract}
and~\ref{fig:connect} are drawn to illustrate, here reproduced on a live run
rather than asserted.

\begin{table}[t]
\centering
\caption{Every accumulator in the running example, grouped by weight. Three
chunks admit exactly three cross-chunk pairs, so a node pair supported by one of
them can take only one of three values; the fourth group is the two-term
combination for an entity named in two paragraphs, and the fifth is same-chunk
co-mention. ``Implied $\cos$'' inverts the value term, $\cos = 1 - 2(1-w)^2$,
which is exact for a single-term pair because its weight \emph{is} its value
term. The first group never becomes an edge: it is below $\theta$, and is kept
un-gated against a document that has not arrived.}
\label{tab:example-weights}
\begin{tabular}{rrrrl}
\toprule
Weight & Support & Pairs & Implied $\cos$ & Evidence \\
\midrule
$0.6958$ & $1$   & $24$ & $0.8149$ & paragraphs $1 \times 2$ --- \emph{below $\theta$} \\
$0.7009$ & $1$   & $32$ & $0.8211$ & paragraphs $1 \times 3$ \\
$0.7171$ & $2$   & $8$  & ---      & both of the above \\
$0.7304$ & $1$   & $48$ & $0.8546$ & paragraphs $2 \times 3$ \\
$1.0000$ & $1,2$ & $59$ & $1.0000$ & same paragraph (co-mention) \\
\bottomrule
\end{tabular}
\end{table}

Table~\ref{tab:example-weights} is the whole derived layer, and it has only five
values in it. That is not a simplification: with three chunks there are three
cross-chunk pairs, so a node pair whose endpoints co-occur in one of them can
carry only that pair's value term. The table also shows both gates working
separately. Twenty-four pairs sit at $0.6958$, below $\theta$, and are held in
the accumulator rather than written --- the latent layer of
Section~\ref{sec:incremental}, at a size one can count. Of the $147$ rows that
\emph{do} clear $\theta$, only $124$ became edges; the per-endpoint cap of $6$
removed the other $23$.

\begin{figure}[t]
\centering
\begin{subfigure}{0.32\textwidth}
  \includegraphics[width=\textwidth]{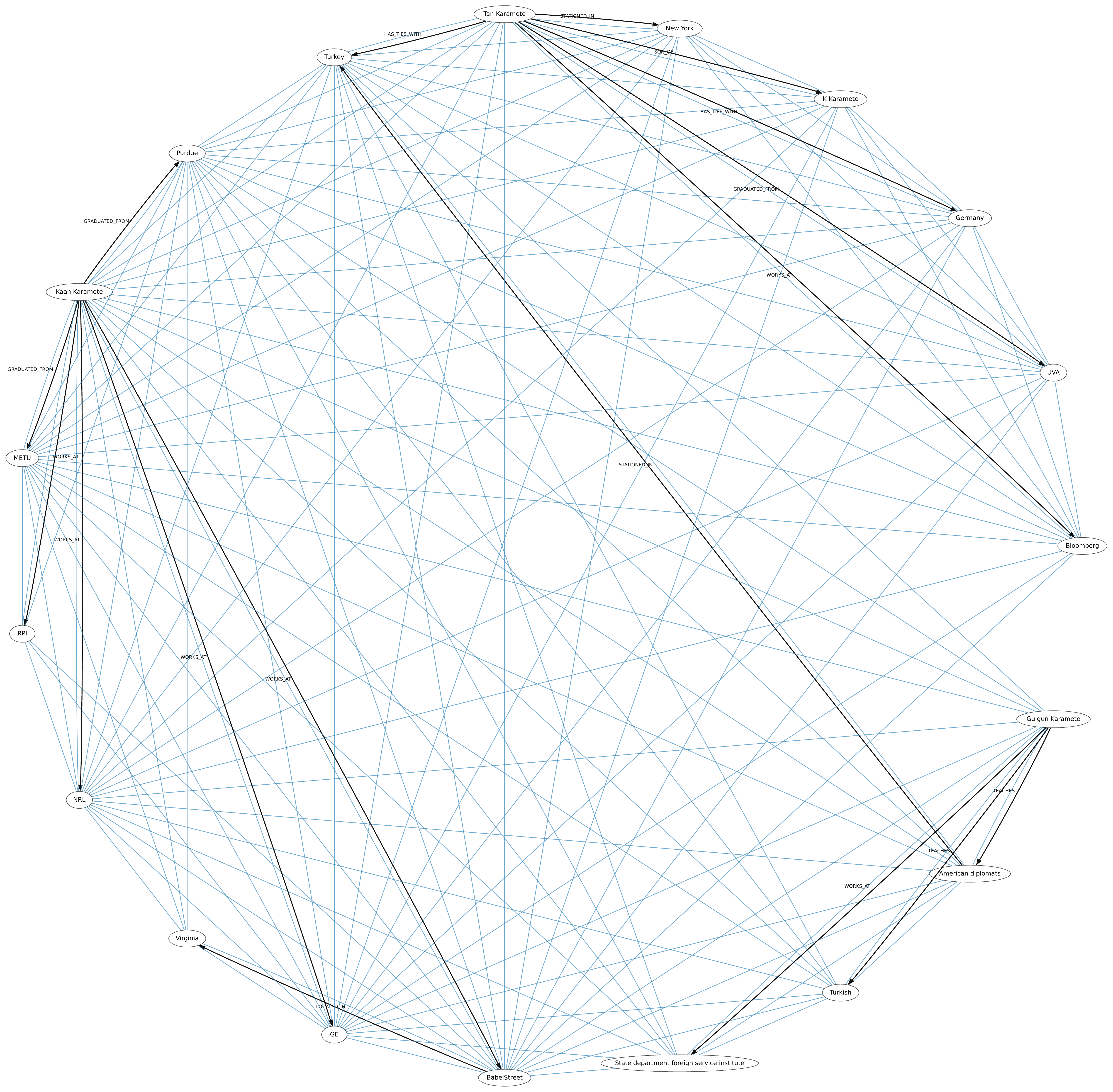}
  \caption{$\theta = 0.70$ (the run's own setting): $124$ derived edges, one
  component.}
  \label{fig:example-theta-a}
\end{subfigure}
\hfill
\begin{subfigure}{0.32\textwidth}
  \includegraphics[width=\textwidth]{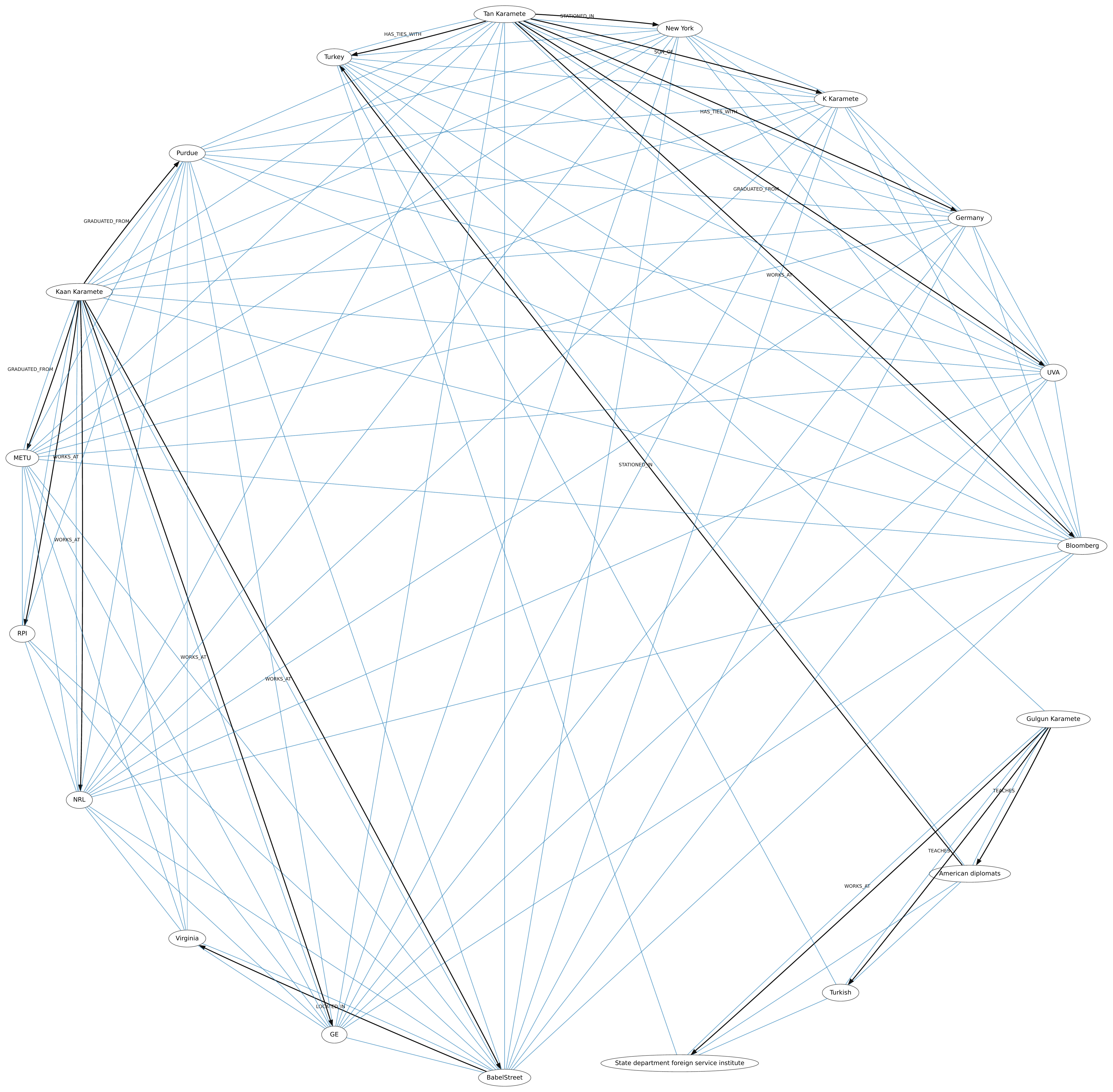}
  \caption{$\theta = 0.72$: $94$ edges. Thinner, but still one component.}
  \label{fig:example-theta-b}
\end{subfigure}
\hfill
\begin{subfigure}{0.32\textwidth}
  \includegraphics[width=\textwidth]{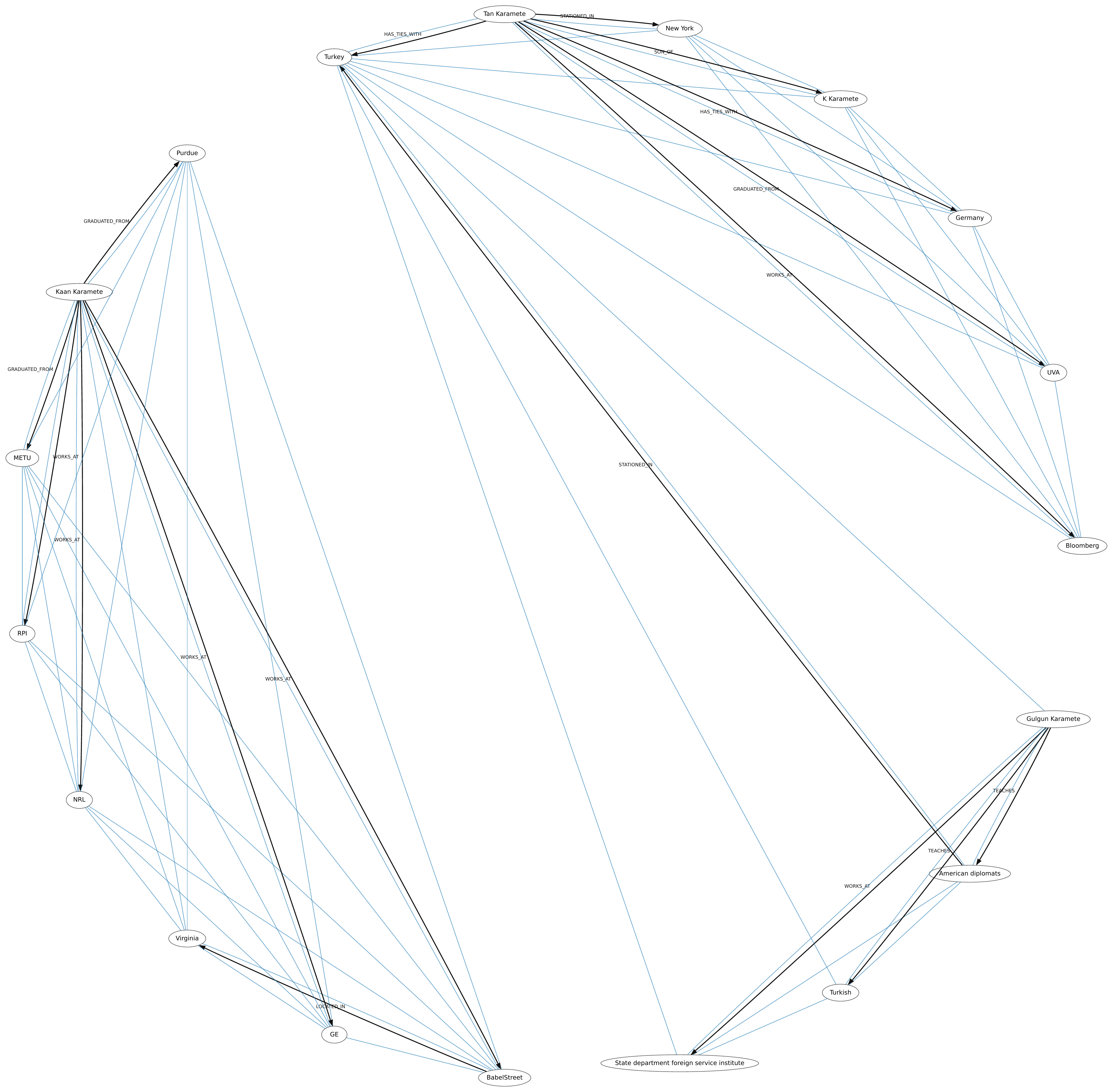}
  \caption{$\theta = 0.74$: $58$ edges, and \emph{two} components again.}
  \label{fig:example-theta-c}
\end{subfigure}
\caption{The running example at three thresholds. Black arrows are the $17$
relations extraction stated and are the same in all three panels; blue arcs are
the derived layer. Sub-threshold edges are drawn invisibly rather than removed,
so the layout is identical in every panel and the reader is comparing thresholds
rather than drawings --- the three images have the same dimensions to the pixel.
The panels are vector art; node names are legible on zoom. The interesting
transition is the third: at $\theta = 0.74$ every cross-paragraph weight of
Table~\ref{tab:example-weights} has been excluded, the $58$ survivors are all
same-paragraph co-mention, and the graph is back to the two components
extraction produced. Tightening the threshold does not merely thin the derived
layer here; past $0.7304$ it switches the pass off.}
\label{fig:example-theta}
\end{figure}

The sweep follows directly. At $\theta = 0.70$ all $124$ edges stand and the
graph is connected. At $0.72$ the $0.7009$ group is excluded and $94$ remain,
still connected, because the strongest cross-paragraph evidence is untouched. At
$0.74$ the last cross-paragraph group goes and the count falls to $58$ --- all of
it same-paragraph co-mention, which by construction cannot connect anything two
paragraphs did not already share --- so the graph separates into the two
components it started with. Between $0.74$ and $1$ nothing further happens. This
is the corpus-scale shape of Section~\ref{subsec:sweep} in miniature: a short
band where $\theta$ decides everything, and a long flat stretch above it where
there is simply no population left for a threshold to act on.

The practical reading is that $\theta$ is not a dial for adjusting how busy a
picture looks. On this corpus every value that matters lies between $0.69$ and
$0.74$, and moving through that band is the difference between the pass making
its claim and not making it at all.

We present all of this as an illustration of the mechanism, not as evidence of
accuracy. Three paragraphs cannot support a claim about precision, and
Section~\ref{sec:limitations} says what would.

\subsection{Embedding width}

Table~\ref{tab:dims} in Section~\ref{sec:embedding} gives the agreement study:
$92\%$ of emitted edges agree with the $3072$-dimensional reference at $768$
dimensions, and $72\%$ at $240$ dimensions, at $960$ bytes per chunk against
$3{,}072$ and $12{,}288$. Because the model is Matryoshka, the narrower vector
is a prefix of the wider one, so this is a pure storage-versus-stability trade
with no retraining and no second embedding pass. The shipped default is $240$;
the vector table is the only structure in the pass that grows without bound with
the corpus, and the layer it feeds is advisory.

\subsection{Query formulation}

Table~\ref{tab:perf} collects the two measurements that changed the
implementation. Both are formulation effects rather than tuning effects: the
same logical operation, expressed differently, at the same scale, and explained
in Section~\ref{subsec:maxby}.

\begin{table}[t]
\centering
\caption{Two formulation results. The top-$k$ measurement is at $768$
dimensions over $100$k rows; the upsert measurement is a $19{,}900$-pair merge.
In both cases the arithmetic performed is identical between the two rows and
only its expression differs.}
\label{tab:perf}
\begin{tabular}{llrr}
\toprule
Stage & Formulation & Time & Ratio \\
\midrule
\multirow{2}{*}{$k$-NN top-$k$}
  & \code{QUALIFY row\_number() OVER (...)} & $35{,}819$\,ms & $25.5\times$ \\
  & \code{max\_by(struct, similarity, k)}   & $1{,}403$\,ms  & $1\times$ \\
\addlinespace
\multirow{2}{*}{accumulate upsert}
  & per-row prepared statement ($\approx\!4$\,ms/row) & $82{,}300$\,ms & --- \\
  & one multi-row \code{VALUES} upsert         & one statement  & --- \\
\bottomrule
\end{tabular}
\end{table}

\subsection{The value term, quantified}

Table~\ref{tab:valueterms} in Section~\ref{sec:idw} gives the two value terms
across the gated range; here is what the change bought in operation. Under the
affine term the gate confined every weight to $[0.825, 1]$, and the threshold
shipped at $0.5$ --- so no setting below $0.825$ would have filtered anything.
Under the chord term the same gate gives a floor of $0.582$, the threshold was
moved to $0.7$, and it now sits above the floor and below the values the strong
cases produce.

Table~\ref{tab:spread} quantifies the rejected alternative: moving the $L_2$ term
into the weight reduces the influence ratio between the closest admissible
neighbour and one at the gate boundary from $350\times$ to $19\times$.

\begin{table}[t]
\centering
\caption{Influence ratio between a very close neighbour ($d = 10^{-3}$) and one
at the gate boundary ($d = \beta = 0.35$), under the shipped weight and under the
Shepard $p=1$ alternative. Inverse-distance weighting exists to make near
evidence dominate; the $p=1$ form gives away more than an order of magnitude of
that dynamic range, and leaves the $\cos = 1$ spike untouched because that spike
is set by $\varepsilon$.}
\label{tab:spread}
\begin{tabular}{lrrr}
\toprule
Weight & $w$ at $d = 10^{-3}$ & $w$ at $d = 0.35$ & Spread \\
\midrule
$w = 1/(\varepsilon + d)$ \quad (Shepard $p=2$, shipped)
  & $999.0$  & $2.86$ & $350\times$ \\
$w' = 1/(\varepsilon + \lVert a-b\rVert)$ \quad (Shepard $p=1$)
  & $22.4$   & $1.20$ & $19\times$ \\
\bottomrule
\end{tabular}
\end{table}

\subsection{Resulting defaults}

Table~\ref{tab:defaults} is the configuration these measurements produce. Every
numeric field is clamped to its stated range rather than rejected; only a
width-or-model conflict against a graph's pinned value raises.

\begin{table}[t]
\centering
\caption{Default configuration and the reason for each value.}
\label{tab:defaults}
\begin{tabular}{llll}
\toprule
Parameter & Default & Range & Rationale \\
\midrule
width          & $240$   & $\{240, 768\}$ & $72\%$ agreement at $0.078\times$ storage \\
$k$            & $20$    & $[1, 200]$     & neighbours per chunk \\
$\beta$        & $0.35$  & $[0, 2]$       & max cosine distance admitted \\
$\theta$       & $0.70$  & $[0, 1]$       & above the $0.582$ floor (Prop.~\ref{prop:floor}) \\
$n$ (cap)      & $10$    & $[1, 1000]$    & per endpoint, union semantics \\
$\varepsilon$  & $10^{-6}$ & fixed        & sets same-chunk dominance \\
keep vectors   & yes     & --- & required to extend the graph cheaply \\
\bottomrule
\end{tabular}
\end{table}

%% file: sections/09-scale.tex
\section{A corpus-scale test on public threat intelligence}
\label{sec:scale}

Three paragraphs show what the pass does; they show nothing about how it behaves
as a corpus grows. Three claims need scale: that the derived layer finds ties
\emph{across} documents rather than restating co-mentions within one, that the
un-gated accumulators leave a latent layer a later document can promote, and
that the ontology stabilises rather than inflating without bound. This section
reports a run over public cyber-threat-intelligence feeds --- $856$ documents
from $43$ live sources, of which $246$ were ingested in an eight-hour run. Every
number below is computed from the run's own per-document log by one script,
\code{tools/tabulate\_run.py}; none is transcribed by hand.

The domain was not chosen for convenience. Threat reporting is written by many
independent publishers about a shared, fast-moving set of actors, malware
families, vulnerabilities and victims: a national CERT, a vendor research team
and a news outlet describe the same campaign within days, in different words,
none citing the others. That is exactly the situation the pass is for --- the tie
between two documents is real, no single document states it, and it becomes
assertable only once both have arrived.

\subsection{Corpus construction}
\label{subsec:corpus}

Documents were built from public RSS and Atom feeds by
\code{tools/fetch\_threat\_feeds.py}. The fetch policy is part of the method, not
an implementation detail. The tool reads \emph{only} the feed documents and never
follows an article link, so the whole corpus costs about fifty HTTP requests and
no question arises about whether retrieving an article body is within a site's
terms. It does not disable TLS verification and does not retry a $403$ behind a
different \code{User-Agent}; a feed that declines is dropped and named in the
output, so a corpus is never quietly smaller than it claims. Four feeds were
dropped this way (an expired certificate, two $403$s and one malformed XML
document), leaving $43$ live sources.

The consequence for the experiment is that a document is a feed summary, not an
article. Summaries are short --- a median of $395$ characters --- so the
paragraph count per document is bimodal: $403$ documents are a single paragraph
and hence a single chunk, while $318$ hit the six-paragraph cap the fetcher
applies. A single-chunk document contributes \emph{no} same-chunk pairs beyond
its own entities' co-mention and must find its neighbours elsewhere in the
graph, which if anything makes the cross-document tier harder to reach, not
easier.

The $246$ documents actually ingested are a prefix of the corpus in feed order,
not a sample drawn from it, and they are not representative of it: their median
length is $693$ characters against the corpus median of $395$, and $133$ of them
sit at the paragraph cap against $65$ at a single paragraph --- the reverse of
the corpus-wide ratio. The direction of that skew matters for how the results
should be read. Longer documents produce more chunks, more chunks produce more
same-chunk pairs, and same-chunk pairs are the \emph{easy} tier. A prefix skewed
long therefore biases the tier split \emph{toward} the within-document tier and
against the cross-document one, so the cross-document share reported in
Table~\ref{tab:tiers} is if anything an underestimate of what the same pipeline
would produce on the whole corpus.

\begin{table}[t]
\centering
\caption{The corpus. Feeds were selected across three categories before the run
and the selection was balanced round-robin so that no single high-volume
publisher dominates; the residual imbalance toward vendor blogs reflects how many
of them publish feeds, not a sampling choice made after seeing results. The last
column is how many of each category the run reached; the shortfall is
budget, not failure, and the seven documents that did fail are accounted for in
Section~\ref{subsec:foldcost}.}
\label{tab:corpus}
\begin{tabular}{lrrrrr}
\toprule
Category & Sources & Documents & Paragraphs & Median chars & Ingested \\
\midrule
Government / CERT     & $3$  & $88$  & $489$  & $316$ & $21$ \\
Vendor research       & $26$ & $573$ & $1{,}636$ & $348$ & $154$ \\
News / aggregators    & $14$ & $195$ & $509$  & $612$ & $71$ \\
\midrule
All                   & $43$ & $856$ & $2{,}634$ & $395$ & $246$ \\
\bottomrule
\end{tabular}
\end{table}

Each document is a JSON record carrying its text, source, category, URL and
paragraph count, and the gateway derives a stable identifier from the text
itself --- \code{doc\_uri = "text:"} followed by the first twelve hex digits of
its SHA-256 --- so the vector store joins back to the corpus offline, without
re-fetching anything and without a side table that can drift.

\subsection{Ingestion, and the stage that actually costs}
\label{subsec:foldcost}

Each document is one \code{/extract} call with the similarity options attached:
$240$ dimensions, $k = 20$, $\beta = 0.35$, $\theta = 0.7$, cap $10$ --- the
shipped defaults of Table~\ref{tab:defaults}. Extraction and folding ran on
Gemini~2.5~Flash. Documents were driven concurrently, which the accumulator
algebra makes safe rather than merely convenient. Within one document the pass
persists its chunks \emph{before} it queries for neighbours. So for two documents
$A$ and $B$ to miss each other, $A$ would have to query before $B$ persisted
\emph{and} $B$ before $A$ persisted, which is a cycle; whichever queries second
sees the other. Concurrency changes wall-clock time and not the result.

Six documents were in flight at a time. The run reached $246$ documents ---
$974$ paragraphs, $7{,}672$ entities, $7{,}787$ relations --- in $8.0$ hours of
wall clock, and $7$ documents failed. Per-document cost is quoted as a
\emph{median}, $91.1$\,s, because the mean is not a useful summary here: it is
$584.6$\,s, six times the median, against a ninetieth percentile of $182$\,s and
one stalled provider request that ran for $6.7$ hours. Summed request time is
$46.9$ hours, so six-way concurrency returned an effective factor of $5.9$. The
failures are worth itemising, because none is a failure of the pass and two are a
lesson. Four were a DuckDB lock conflict: the metadata store is a single file, a
test run happened to open it, and the colliding documents were lost. One was a
read timeout. The last two were a payload shape --- the extractor is asked for
facets as \code{\{name, axis\}} objects and roughly one document in forty returns
a bare list of strings, which every downstream reader then indexed as a
dictionary. A failed document writes no ledger row and no vectors, so the graph
stays consistent and the document can be re-driven; but with a generative
extractor the schema is a request, not a guarantee, and consumers have to coerce
rather than assume.

Running the corpus surfaced a scaling property of the surrounding pipeline that
is worth reporting because it is counter-intuitive and because it is not in the
similarity pass at all. Per-document time climbed from $11.9$\,s on the first
document to $80.4$\,s by the twelfth, on documents of the same size. What grows
with the corpus is the \emph{ontology}, and the only stage that reads the whole
ontology is label folding: each previously-unseen label costs one LLM round-trip
that asks whether it is a synonym of an existing canonical, and those round-trips
were issued one at a time. Profiling a six-paragraph document against an
established ontology gave extraction $72$\,s, folding $387$\,s, and the entire
similarity pass --- embed, $k$-NN, expand, accumulate, select, link --- $3$\,s.
Folding was $84\%$ of the wall clock. The control that attributes it is the same
document ingested into an \emph{empty} graph, where folding takes $2.0$\,s
because with no canonicals to fold into every check short-circuits before it
reaches the model.

\begin{table}[t]
\centering
\caption{Where an extraction's time goes, on one six-paragraph document at $240$
dimensions. The middle column is the same document into an empty graph, which is
what identifies the cost: nothing about the document changed, only the size of
the ontology it is folded against. The right column is after fanning the
independent fold-checks out across eight workers and sending the document's
paragraphs to the extractor in parallel.}
\label{tab:foldcost}
\begin{tabular}{lrrr}
\toprule
Stage & Established ontology & Empty graph & After fan-out \\
\midrule
Extraction           & $72.0$\,s  & $112.5$\,s$^{\dagger}$ & $21$\,s \\
Label folding        & $387.0$\,s & $2.0$\,s          & $49$\,s \\
Similarity pass      & $3.2$\,s   & $3.0$\,s          & $3$\,s \\
\midrule
Total                & $463.2$\,s & $117.5$\,s        & $73.3$\,s \\
\bottomrule
\end{tabular}

\smallskip
{\footnotesize $^{\dagger}$ by difference; only the folding stage and the total
were instrumented in the empty-graph control.}
\end{table}

The fix is fan-out, not caching: the fold-checks for distinct labels are
independent --- each reads a snapshot of the canonical set taken before the walk
and writes nothing --- so they are issued concurrently, while every write back to
the metadata store stays on the calling thread in the original walk order, so the
store sees exactly the sequence it saw before. That is a $6.3\times$
improvement end to end, and it is reported here for two reasons. First, it is the
honest accounting of what a document costs in this pipeline: the pass this
article is about is $3$\,s of a $73$\,s document, and a reader deciding whether
to adopt it should know that it is not the expensive part. Second, it is a
property of \emph{ontology folding} specifically --- the cost grows with what has
already been learned --- and any system that maintains a canonical vocabulary
across a corpus will meet it.

\subsection{The ontology as it is learned}
\label{subsec:ontology-evolution}

Nothing in the pipeline is given a schema. The first document arrives against an
empty graph, the extractor proposes whatever types its text seems to need, and
every document after that is folded against what the previous ones left behind.
The ontology in Figure~\ref{fig:ontology-early} is therefore not a design; it is
a residue.

\begin{figure}[p]
\centering
\begin{subfigure}{0.66\textwidth}
  \centering
  \includegraphics[width=\textwidth]{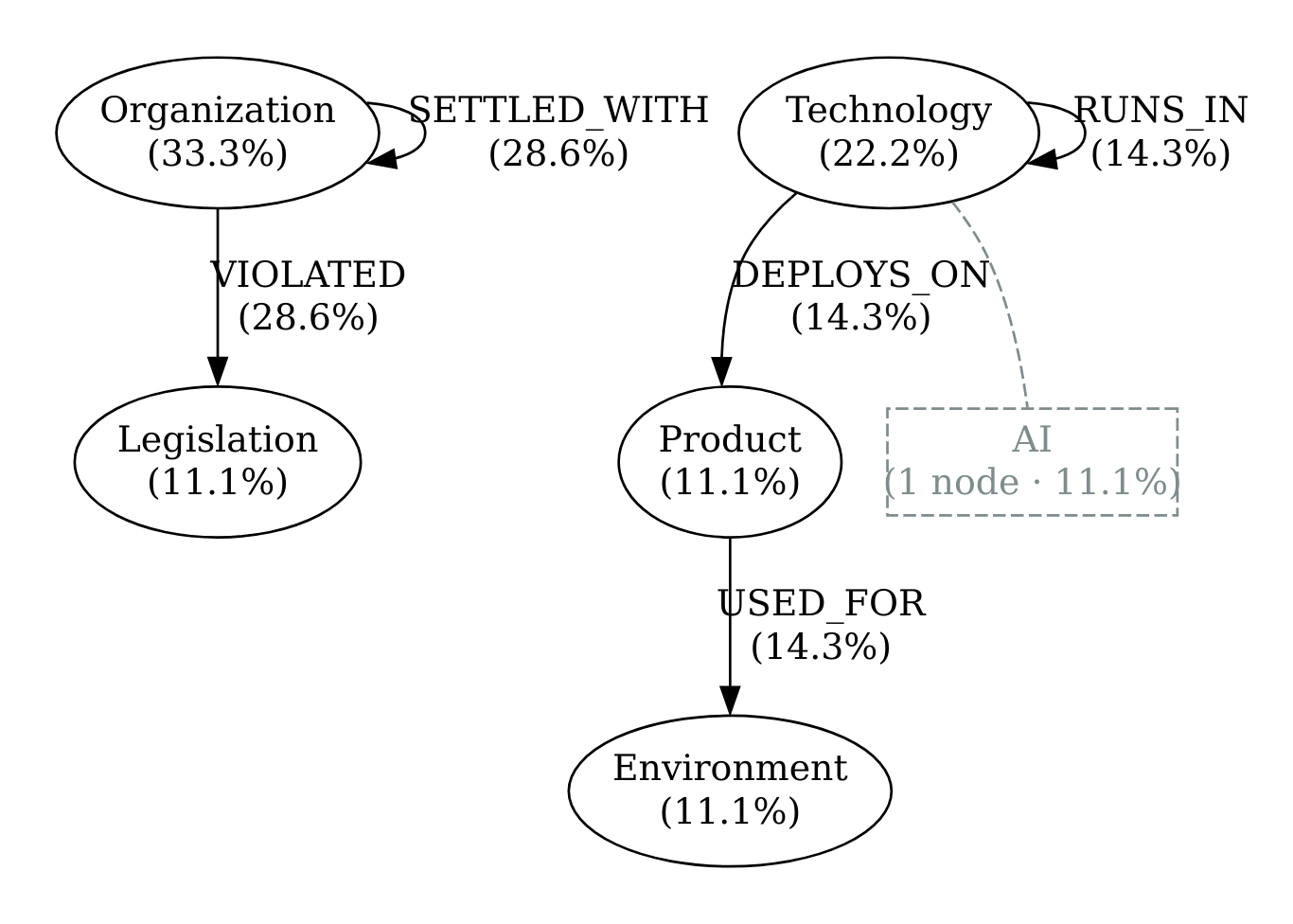}
  \caption{After one document: $5$ types, $5$ relation types, $5$ arrows over
  $7$ edges.}
  \label{fig:ontology-1}
\end{subfigure}

\vspace{10pt}
\begin{subfigure}{\textwidth}
  \centering
  \includegraphics[width=\textwidth]{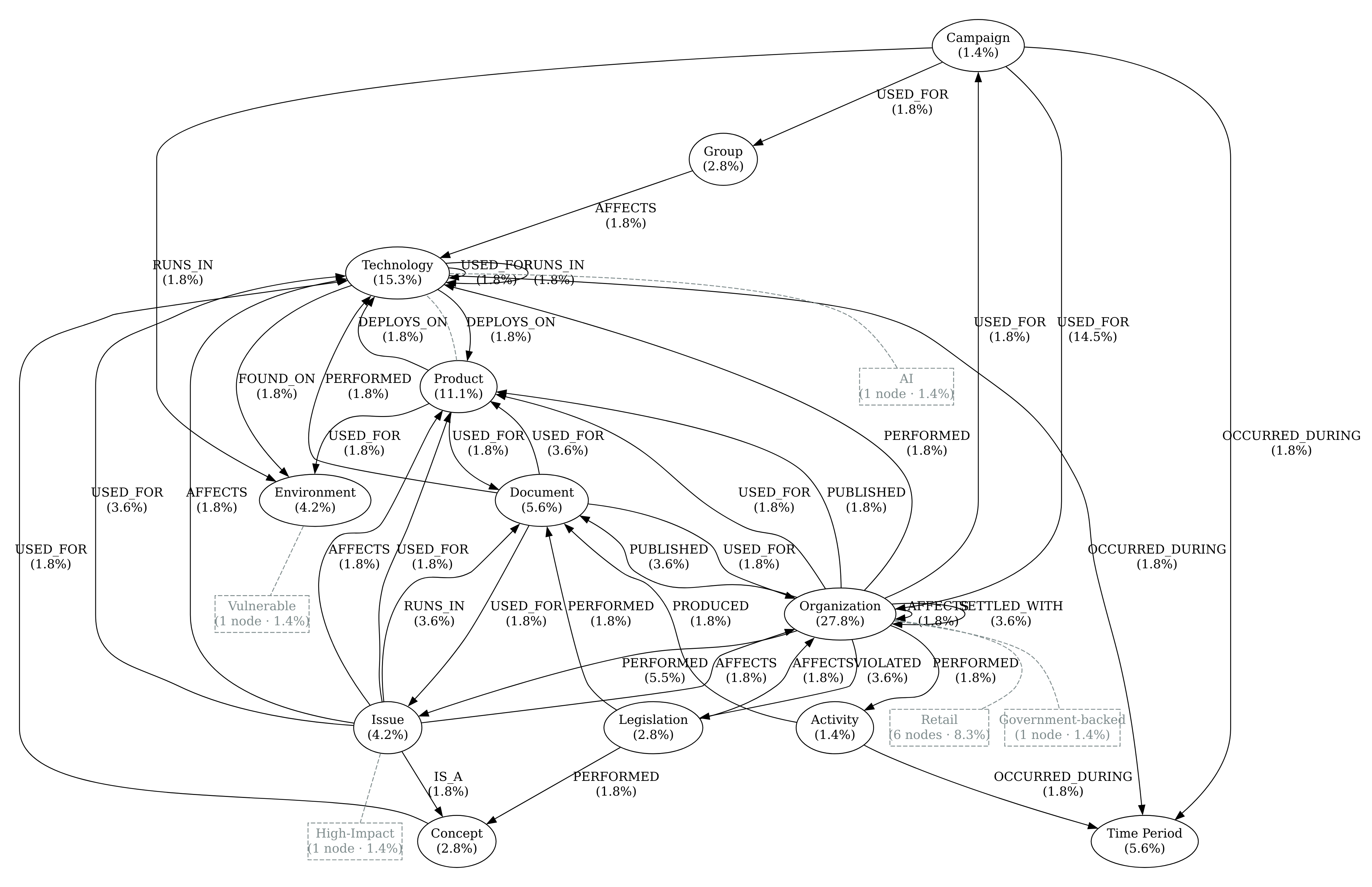}
  \caption{After five: $13$ types, $12$ relation types, $40$ arrows over $55$
  edges. Nothing from (a) was discarded.}
  \label{fig:ontology-5}
\end{subfigure}
\caption{The extracted ontology after one and after five documents, rendered by
the workbench itself --- the gateway's own schema-DOT builder laid out by the
same \code{@hpcc-js/wasm} Graphviz build the browser loads. Solid ellipses are
structural types and carry their share of all nodes; solid arrows are relation
types and carry that arrow's share of all edges. Dashed boxes are \emph{facets}
--- a dimension \emph{of} a type rather than a relationship between two, so
\code{Social Media} hangs off \code{Organization} instead of being a second
organization-like thing next to it. Both panels are reconstructed after the fact
from the \code{first\_seen\_ts} every node and edge carries; the run itself was
not instrumented.}
\label{fig:ontology-early}
\end{figure}

Two things in the five-document panel are worth naming, because they are the
mechanism the rest of this section measures at scale. The first is that
\code{Organization} is already a hub at $27.8\%$ of nodes while \code{Campaign}
and \code{Activity}, at $1.4\%$ each, are still one-node types waiting to see
whether anything else in the corpus needs them. A type's share is not a
property of the type; it is a property of what has arrived so far, and the
long tail here is a queue, not a residue.

The second is what is \emph{not} in the panel. There is one \code{Product} and
one \code{Technology}, not the pairs \code{Product}/\code{Software Product} and
\code{Technology}/\code{TechnologyConcept} an earlier run of this same corpus
produced at the same point. Those near-synonyms were not a delay that more
documents would have resolved. We ran that corpus far enough to establish they
were permanent: it ended at $903$ node and $1{,}315$ edge labels, too large for
the workbench's own Graphviz build to lay out. Four causes were at work, and the
decisive one was the cheapest to fix. The fold was being asked to merge names
\emph{after} the fact, when the extractor could have been told the existing
vocabulary and asked to reuse it in the first place. Reuse decided in the
extraction prompt is free; the same decision made afterwards costs one model
round-trip per label and landed less than half the time. The vocabulary is passed
as a strong preference rather than a closed enum, because document $1$ has no
vocabulary and document $500$ must still be able to name a genuinely new kind.
We record the earlier failure rather than quietly re-running, because ``the
ontology converges'' is only worth anything alongside the configuration under
which it did not.

\subsection{The same ontology with the derived layer put back}
\label{subsec:derived-schema}

Both panels of Figure~\ref{fig:ontology-early} are the \emph{extracted} ontology:
every arrow in them is a relation some document states, and \EMB{} is filtered
out. That is the right default for a figure about what extraction produces ---
drawing the derived layer into it would credit the extractor with something it
never found --- but it means the layer this article is about does not appear in
either picture. Figure~\ref{fig:ontology-embedded} is the five-document panel
with the filter removed.

\begin{figure}[t]
\centering
\includegraphics[width=\textwidth]{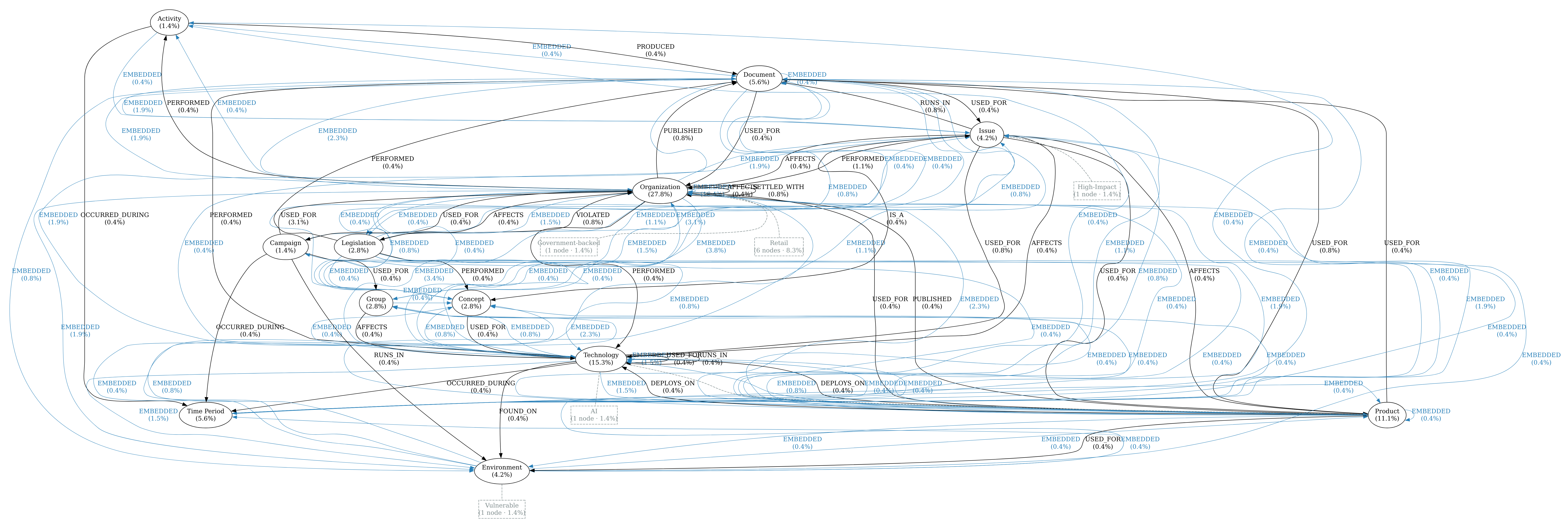}
\caption{The five-document ontology of Figure~\ref{fig:ontology-5} with the
derived layer put back, at the same cutoff and from the same builder. Blue
arrows are \EMB{}; black arrows and ellipses are exactly the extracted ontology
of the earlier panel. The structural vocabulary does not move --- the pass mints
no types --- but the schema does: $12$ relation types over $40$ arrows and $55$
edges become $13$ over $107$ and $262$, so the similarity layer contributes $67$
arrows and $207$ edges against extraction's $55$. Nearly four derived edges for
every stated one, after five documents. The derived arrows are drawn off
Graphviz's constraint graph, so the black skeleton keeps the layout it has in
Figure~\ref{fig:ontology-5} and the two can be read against each other; letting
$67$ arrows participate in rank assignment buries the skeleton the reader came
for. Parameters are the run's throughout: $\beta = 0.35$, $\theta = 0.7$,
$k = 20$, $n = 10$, $240$ dimensions. Two cautions come with the density. At type
level the layer approaches saturation, so what makes it useful is not which type
pairs it touches but how strongly --- and a schema view deliberately discards
exactly that. And at \emph{this} cutoff the blue arrows are not graded at all:
every one of the $207$ edges has weight $\geq 0.9997$, because five documents have
produced no cross-document evidence yet to dilute a co-mention. No threshold below
$1$ removes any of them. Section~\ref{subsec:sweep} sweeps $\theta$ over the full
corpus, where the two levels do come apart.}
\label{fig:ontology-embedded}
\end{figure}

\subsection{What the vocabulary did over 246 documents}
\label{subsec:vocab}

Table~\ref{tab:vocab} is the cumulative count of distinct canonical types as
documents arrive. The structural vocabulary saturates: the twenty-fourth entity
type is reached at document $246$ and the rate over the last $46$ documents is
$0.04$ new types per document, against $1.33$ over the first ten. No fold in the
run was \emph{forced} --- the hard cap of $100$ canonicals per kind, at which the
fold-check is re-asked as a choice with no null option, was never approached.

\begin{table}[t]
\centering
\caption{Cumulative distinct canonical types against documents ingested, counted
from the driver's per-document log --- a type enters at the first document
\emph{labelled} with it. The right-hand column is the marginal rate over the
interval ending at that row, which is the quantity that decides whether an
ontology converges; a count that merely grows more slowly than the corpus still
grows without bound. Figure~\ref{fig:growth} counts the same vocabulary from the
ledger instead, at the moment a canonical is \emph{minted}, and ends four lower.}
\label{tab:vocab}
\begin{tabular}{rrrr}
\toprule
Documents & Entity types & Relation types & New entity types / doc \\
\midrule
$1$   & $3$  & $3$  & $3.00$ \\
$10$  & $15$ & $16$ & $1.33$ \\
$25$  & $17$ & $17$ & $0.13$ \\
$50$  & $17$ & $19$ & $0.00$ \\
$100$ & $19$ & $24$ & $0.04$ \\
$200$ & $22$ & $39$ & $0.03$ \\
$246$ & $24$ & $43$ & $0.04$ \\
\bottomrule
\end{tabular}
\end{table}

\begin{figure}[t]
\centering
\includegraphics[width=\textwidth]{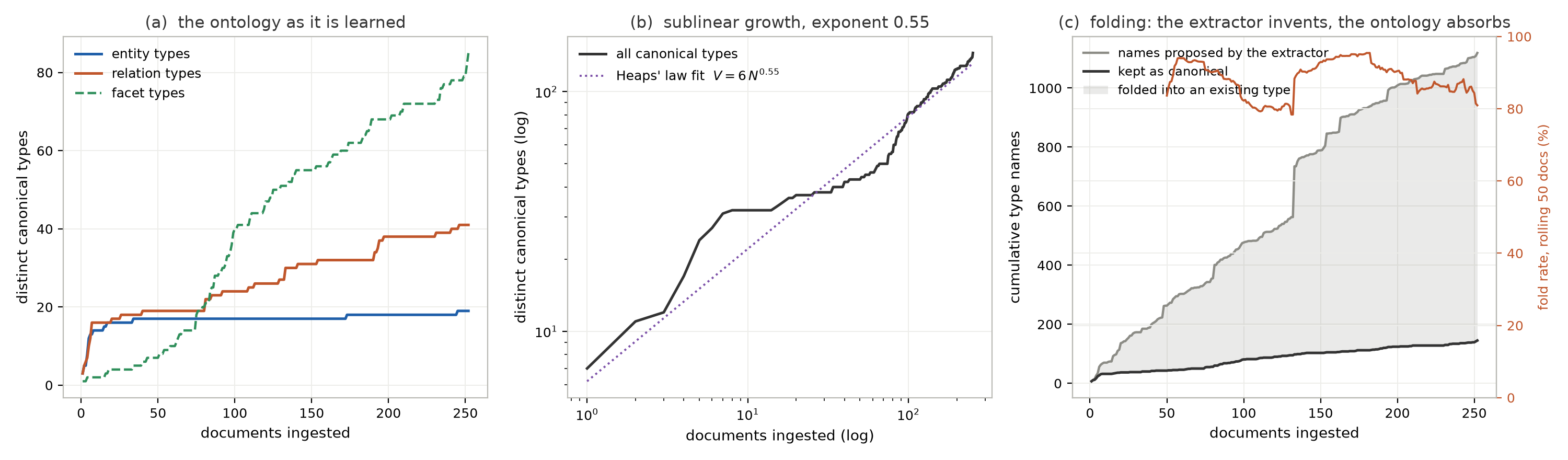}
\caption{The ontology over the run, reconstructed from the \code{first\_seen\_ts}
each canonical carries; the run itself was not instrumented. (a) Entity types
flatten at $19$ from about document $80$ and do not move again, relation types
climb more slowly to $41$, and \emph{facet} names --- a dimension of a type
rather than a type --- rise to $85$ and are still rising at the end of the run.
That divergence is the section's main negative result. (b) The same total against
a Heaps' law fit; vocabulary growth in natural language is sublinear, and an
exponent of $0.55$ says an extracted ontology behaves the same way. (c) What the
fold is doing: the extractor proposed $1{,}119$ names of which $974$ ($87.0\%$)
were absorbed into an existing canonical, and the rolling fold rate stays between
$80\%$ and $95\%$ rather than declining --- an ontology that is saturating, since
an increasing share of what the extractor invents is something it already has a
word for.}
\label{fig:growth}
\end{figure}

That is the claim the section set out to test, and Figure~\ref{fig:growth}(a)
and (b) support it --- for structural types. Panel (c) is the mechanism behind
the flat line: nearly nine names in ten that the extractor proposes are absorbed
rather than minted, and $145$ canonicals survive out of $1{,}119$ proposals.
Table~\ref{tab:vocab} counts a type from the moment a document is labelled with
it and so reaches $24$, four more than the ledger's $19$; the difference is names
first minted on a facet axis and later used structurally, and the two counts
bracket the same flat curve.

The claim does \emph{not} hold for the graph's node labels, and the difference is
worth being precise about, because the honest number is the larger one. A node
carries a structural type \emph{and} its facets as labels, so the graph at the end
of the run has $101$ distinct node labels, not $19$. Facets fold onto an
\emph{axis}, but their names are not capped the way a type's are: the ledger holds
$104$ canonical entity names in all, of which $85$ were minted on a facet axis.
Some are real --- \code{Agentic}, \code{Ephemeral}, \code{Cloud-Hosted},
\code{Export-Controlled}. Others are adjectives lifted out of prose and are facets
of nothing: \code{Correct}, \code{Different}, \code{Final}, \code{Frustrating}.

Measuring this turned up one thing that \emph{was} a defect, and it is worth
separating from the rest because the two look alike and only one is fixable. A
facet folds onto an \emph{axis} --- the dimension it measures --- and the two
placeholders \code{EntityType} and \code{RelationType} mean \emph{no axis given},
which is to say the structural axis itself. They were deliberately never folded,
so that a label arriving without an axis could not be filed under whichever real
axis happened to rank closest. What that missed is that the extractor proposes an
axis literally named \code{Type}, and \code{TYPE} for relations, and those are the
same axis under another spelling. Unfolded, they minted a \emph{facet} axis
duplicating the structural role: $271$ of the ledger's $1{,}119$ type rows, $24\%$,
carried a structural axis under a facet name, and each name on it was then counted
as a facet of nothing. The fold now recognises them --- an axis is the structural
one when all of its word stems come from \{\code{entity}, \code{relation},
\code{node}, \code{edge}, \code{type}\} and one of them is \code{type}, so
\code{Type} and \code{EntityType} resolve while \code{Technology Type} stays a
facet dimension --- and the decision replays over an existing ledger by string
matching alone, with no re-extraction. Because an axis is a grouping and not a
label, no node label moves, and every count in this section is unchanged by the
repair.

Two things that look like the same defect are not, and are left alone.
\code{Impact} and \code{IMPACT} do coexist, but in two separate registries
following the graph's two naming conventions --- Title Case for entity labels,
\code{UPPER\_SNAKE} for relation types --- and within either registry no two
canonicals differ only by case. And $18$ canonicals sit on both the structural and
a facet axis, which is the domain rather than an error: \code{Retail} genuinely is
both a kind of thing and a dimension of one.

The uncapped facet vocabulary we report rather than fix, because the fix is an
ontology-design decision. The cap that disciplines types counts canonicals per
kind, and a facet is a different kind --- a dimension \emph{of} a type, of which
a domain may legitimately have many. Capping facet names would work and would
change what the ontology means. What the measurement establishes is narrower and
still useful: the mechanism that stabilises a vocabulary works where it is
applied, and where it is not --- the dashed curve in
Figure~\ref{fig:growth}(a) --- the vocabulary does not stabilise on its own.

\subsection{The derived layer at corpus scale}
\label{subsec:derived}

The run wrote $635{,}562$ \code{EMBEDDED} emissions, which resolve to $69{,}867$
distinct derived edges in a graph of $6{,}706$ nodes and $77{,}999$ edges. The
derived layer is thus $89.6\%$ of all edges --- the extracted graph is the small
part, which is what one should expect from a pass that considers every pair of
entities that share a neighbourhood rather than only pairs a sentence relates.

\begin{table}[t]
\centering
\caption{The derived layer by provenance tier. ``Written'' is edges above
$\theta = 0.7$; the two right-hand columns are the accumulator rows behind them,
which are persisted \emph{un-gated}. The bottom-right cell is the latent layer
this article argues for: $92{,}464$ cross-document pairs that a later document
can promote by raising their accumulated weight, with no recomputation of
anything already stored.}
\label{tab:tiers}
\begin{tabular}{lrrrrr}
\toprule
Tier & Written & Share & Median $w$ & Accum.\ $\ge\theta$ & Accum.\ $<\theta$ \\
\midrule
Same chunk     & $46{,}215$ & $66.1\%$ & $1.000$ & $154{,}022$   & $0$ \\
Same document  & $252$      & $0.4\%$  & $0.761$ & $65{,}010$    & $858$ \\
Cross document & $23{,}400$ & $33.5\%$ & $0.749$ & $1{,}205{,}441$ & $92{,}464$ \\
\midrule
All            & $69{,}867$ & & $1.000$ & $1{,}424{,}473$ & $93{,}322$ \\
\bottomrule
\end{tabular}
\end{table}

Three readings of Table~\ref{tab:tiers} matter. First, the same-chunk tier
dominates the written edges at $66.1\%$ and sits at a median weight of exactly
$1.000$, which is by construction: two entities in one chunk enter as a self-pair
at $\cos = 1$ and no amount of corpus can outweigh that. Anyone reading a derived
edge's weight should know that a weight of $1$ means co-mention and nothing more.
Second, and against that, a third of the written layer --- $23{,}400$ edges --- is
\emph{cross-document}, at a median weight of $0.749$. Those are the ties no single
document states, and they exist in this quantity only because the corpus is many
publishers writing about a shared subject. Third, the un-gated accumulator table
is $1{,}517{,}795$ rows, or $6.8\%$ of all node pairs, of which $93{,}322$ sit
below the gate. That is the evolving part: it is not a cache and not a to-do
list, it is a set of assertions the corpus currently does not support strongly
enough to write down, held in a form where the next document can change that
verdict arithmetically.

One property does not survive contact with a corpus, and it is the per-endpoint
cap. \code{top\_n} $= 10$ bounds how many neighbours an endpoint keeps
\emph{per run}, not per graph, so a hub accumulates a fresh top ten on every
document that mentions it. Over $246$ documents the median derived degree is
$11$, the ninetieth percentile is $36$, and the maximum is $742$; $3{,}689$ of
$6{,}668$ endpoints ($55.3\%$) exceed the nominal cap. This is a real design
consequence rather than a bug --- a per-graph cap would require reconsidering an
endpoint's whole neighbourhood on every write, which is exactly the recomputation
the incremental design exists to avoid --- but a reader should not take
\code{top\_n} as a degree bound. The weight distribution is otherwise as the
arithmetic requires: minimum $0.707$, tenth percentile $0.741$, median $1.000$,
and \emph{zero} edges below the floor of $1 - \sqrt{\beta/2} = 0.582$ established
in Section~\ref{sec:idw}. Support --- the number of chunk pairs behind an edge
--- has median $1$, ninetieth percentile $3$ and maximum $119$.

\subsection{Tightening \texorpdfstring{$\theta$}{theta} thins the instances, not the ontology}
\label{subsec:sweep}

Figure~\ref{fig:ontology-embedded} raises an obvious question: the derived layer
looks saturated at type level, so would a stricter threshold thin it back to
something legible? Table~\ref{tab:sweep} answers it, and the answer is no --- for
a reason worth stating, because it bounds what $\theta$ is good for.

The sweep costs nothing to run, which is itself a consequence of the design. The
accumulators are persisted un-gated (Section~\ref{sec:incremental}) and a written
edge carries its weight as a property, so re-applying $\theta$ at read time needs
neither the corpus nor the embedder. Two populations are swept: the $69{,}867$
written edges, which can only be tightened because they were already gated at
$\theta = 0.7$; and the $1{,}517{,}795$ accumulator rows, which are every pair the
pass has ever had evidence for and so sweep from the floor up. Each is counted at
both levels --- instance edges, and the directed structural-label arrows an
ontology diagram draws.

\begin{table}[!htb]
\centering
\caption{The derived layer against the weight threshold, re-applied at read time.
``Arrows'' is distinct \code{src}\,$\to$\,\code{dst} structural-type pairs --- an
edge of Figure~\ref{fig:ontology-embedded}, drawn at corpus scale. Both
populations lose instances fast and arrows slowly. Below $\theta = 0.70$ the
written column cannot move, since that is the gate the run applied; the
accumulator column can, and shows the latent layer underneath it.}
\label{tab:sweep}
\begin{tabular}{lrrrr}
\toprule
& \multicolumn{2}{c}{Written edges} & \multicolumn{2}{c}{Un-gated accumulators} \\
\cmidrule(lr){2-3}\cmidrule(lr){4-5}
$\theta$ & Instances & Arrows & Pairs & Arrows \\
\midrule
$0.582$ & $69{,}867$ & $347$ & $1{,}517{,}795$ & $464$ \\
$0.70$  & $69{,}867$ & $347$ & $1{,}424{,}473$ & $464$ \\
$0.75$  & $57{,}540$ & $341$ & $394{,}668$     & $415$ \\
$0.80$  & $46{,}928$ & $319$ & $159{,}969$     & $336$ \\
$0.90$  & $46{,}687$ & $317$ & $154{,}505$     & $323$ \\
$1.00$  & $32{,}817$ & $285$ & $130{,}942$     & $294$ \\
\bottomrule
\end{tabular}
\end{table}

Moving $\theta$ from $0.70$ to $0.80$ discards $88.8\%$ of the candidate pairs and
$27.6\%$ of the schema arrows. The gap is structural rather than a property of
this corpus: an arrow exists if \emph{any} instance under it survives, so the type
level is a disjunction over thousands of instances, and a disjunction is very hard
to switch off --- each arrow only ever needed one. $\theta$ is therefore a strong
instance-level control and a weak ontology-level one. Anyone hoping to declutter a
schema diagram by raising it will find they have deleted most of the graph and
almost none of the picture; the lever that works on the picture is drawing arrows
by strength, which is what Figure~\ref{fig:ontology-embedded} gives up by being a
schema view.

The flat stretch from $0.80$ to $0.95$ is not a coarse grid. The weight
distribution is bimodal, and sharply so: of $1.5$ million pairs, only $5{,}513$
--- $0.36\%$ --- have a weight anywhere in $[0.80, 0.999)$. The $23{,}514$ more
between $0.999$ and $1$ are co-mention carrying a floating-point residue rather
than a third mode. Evidence is either same-chunk co-mention pinned at $1.0$ or
cross-document material in a band just above the $0.582$ floor of
Section~\ref{sec:idw}, with next to nothing in between, so there is no population
for a threshold in that range to act on. This is the corpus-scale form of what
Figure~\ref{fig:ontology-embedded} shows at five documents, where the second mode
does not exist yet and the sweep is flat all the way to $1$.

\begin{figure}[p]
\centering
\begin{subfigure}{0.86\textwidth}
  \centering
  \includegraphics[width=\textwidth]{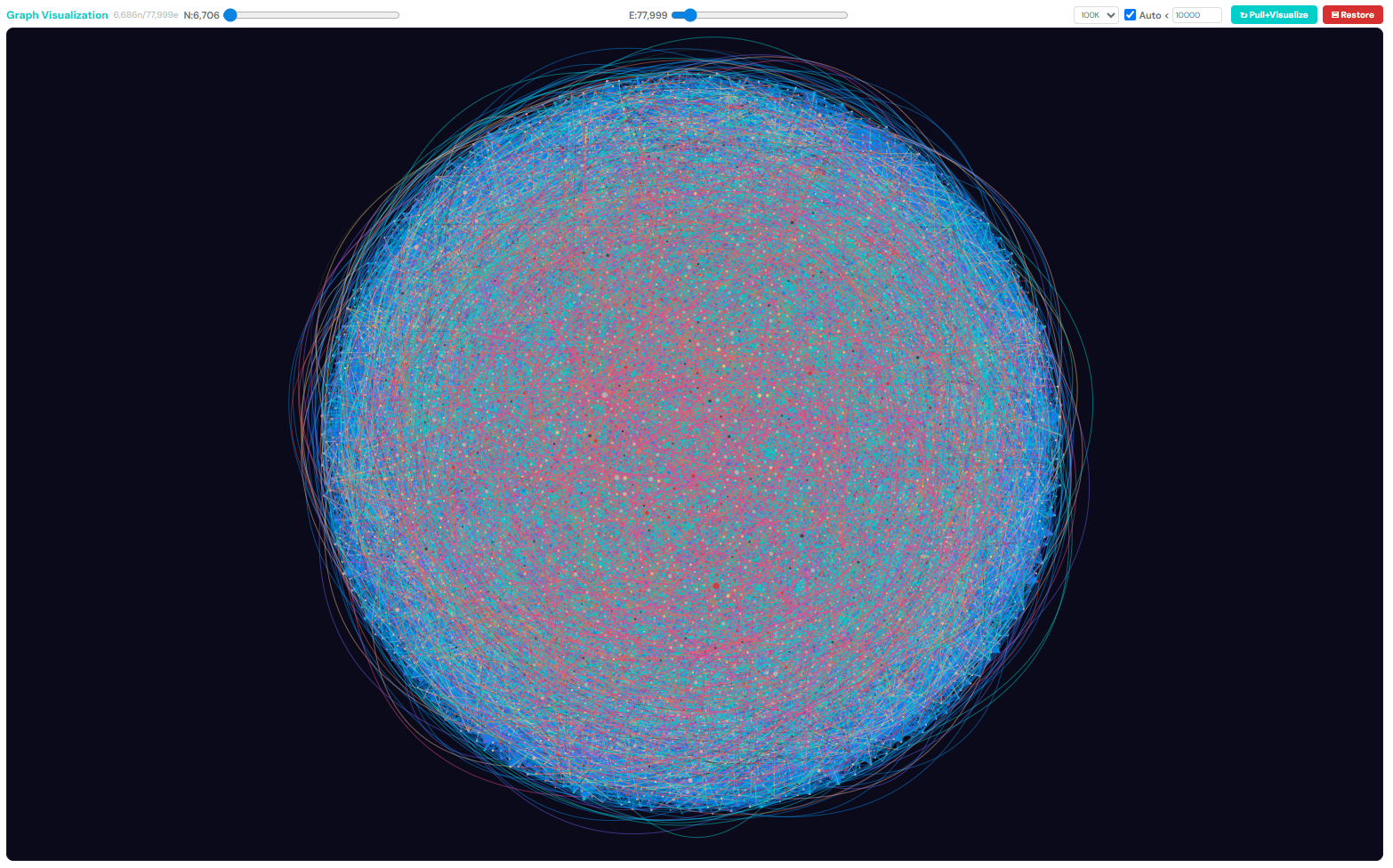}
  \caption{Everything: $77{,}999$ edges over the $6{,}686$ nodes that carry one.}
  \label{fig:canvas-all}
\end{subfigure}

\vspace{4pt}

\begin{subfigure}{0.86\textwidth}
  \centering
  \includegraphics[width=\textwidth]{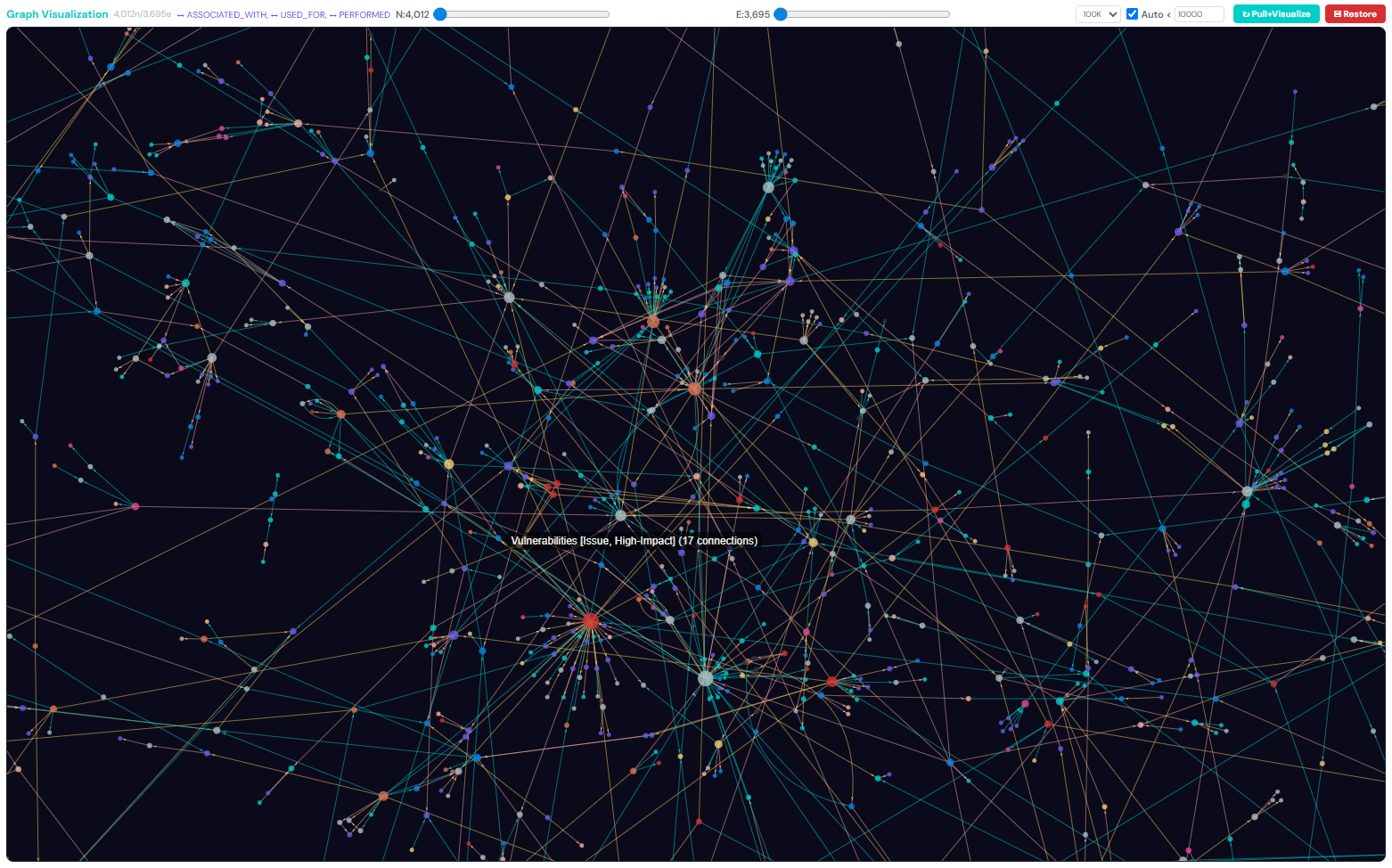}
  \caption{\code{EMBEDDED} deselected: the three largest extracted relation types
  --- \code{ASSOCIATED\_WITH} ($1{,}653$), \code{USED\_FOR} ($1{,}068$) and
  \code{PERFORMED} ($974$) --- $3{,}695$ edges over $4{,}012$ nodes.}
  \label{fig:canvas-extracted}
\end{subfigure}
\caption{The same graph with and without its derived layer, drawn by a force
simulation in the browser. \subref{fig:canvas-all} is what $89.6\%$ derived
\emph{looks} like: an undifferentiated disc in which the layout has nothing left
to separate, because almost every pair of entities sharing a neighbourhood now
shares an edge. It is not a pretty picture of a graph, and that is the finding.
\subref{fig:canvas-extracted} is the extracted graph underneath, reached by
deselecting one edge label in the workbench's Edges-by-Label panel: hub-and-spoke
stars around frequently-named actors and technologies, joined by long
\code{ASSOCIATED\_WITH} spans, with a fringe of small components that are
individual documents the corpus has not yet tied to anything. The hovered node is
\code{Vulnerabilities}, and it gives this subsection's argument at the scale of
one node --- the tooltip's ``$17$ connections'' is its degree \emph{within the
current selection}, its full extracted degree is $47$, and it carries $302$
\code{EMBEDDED} edges, $86.5\%$ derived.}
\label{fig:canvas}
\end{figure}

Figure~\ref{fig:canvas} draws the same layer rather than tabulating it, and it
makes visible one thing a table cannot: the derived edges are what connect the
graph, and they connect it too well. Extraction produces a document-shaped
scatter, because nothing in a single feed summary relates its entities to any
other document's --- panel~\subref{fig:canvas-extracted} is that scatter, still
legible, with per-document components one can count by eye. Adding the derived
layer pulls all of it into the single body of panel~\subref{fig:canvas-all}. That
is the result the pass was built to produce and simultaneously the reason it
needs a gate. Whether the object is \emph{useful} at this density is a separate
question, and Section~\ref{sec:limitations} takes it up. The immediate practical
point is smaller: the derived layer must be a \emph{toggleable} layer in whatever
tool inspects the graph, because a single view of both is a view of neither.

\subsection{Seeing the neighbourhood structure: the corpus on $S^2$}
\label{subsec:sphere}

Everything in Section~\ref{sec:idw} is arithmetic on cosines, which is hard to
disbelieve and equally hard to believe. This subsection turns it into a picture
--- and, because a picture of clustering proves nothing on its own, into the two
numbers that say whether the picture is entitled to its shape.

The vectors are $L^2$-normalised at write time, so they already lie on $S^{239}$.
Flattening them is not a change of geometry but a choice of how many coordinates
of an orthonormal frame to keep: two give a circle, three give a sphere. The
frame that preserves the most squared distance is PCA, and for unit vectors that
is not one option among several. The matrix of pairwise cosines \emph{is} the
Gram matrix, and classical multidimensional scaling on a Gram matrix is
algebraically identical to PCA on the vectors. The vectors are centred first:
sentence embeddings are anisotropic and every one carries a large shared
component, so the uncentred first component is ``is this text at all'' --- true
of every point and separating none.

\begin{figure}[t]
\centering
\includegraphics[width=\textwidth]{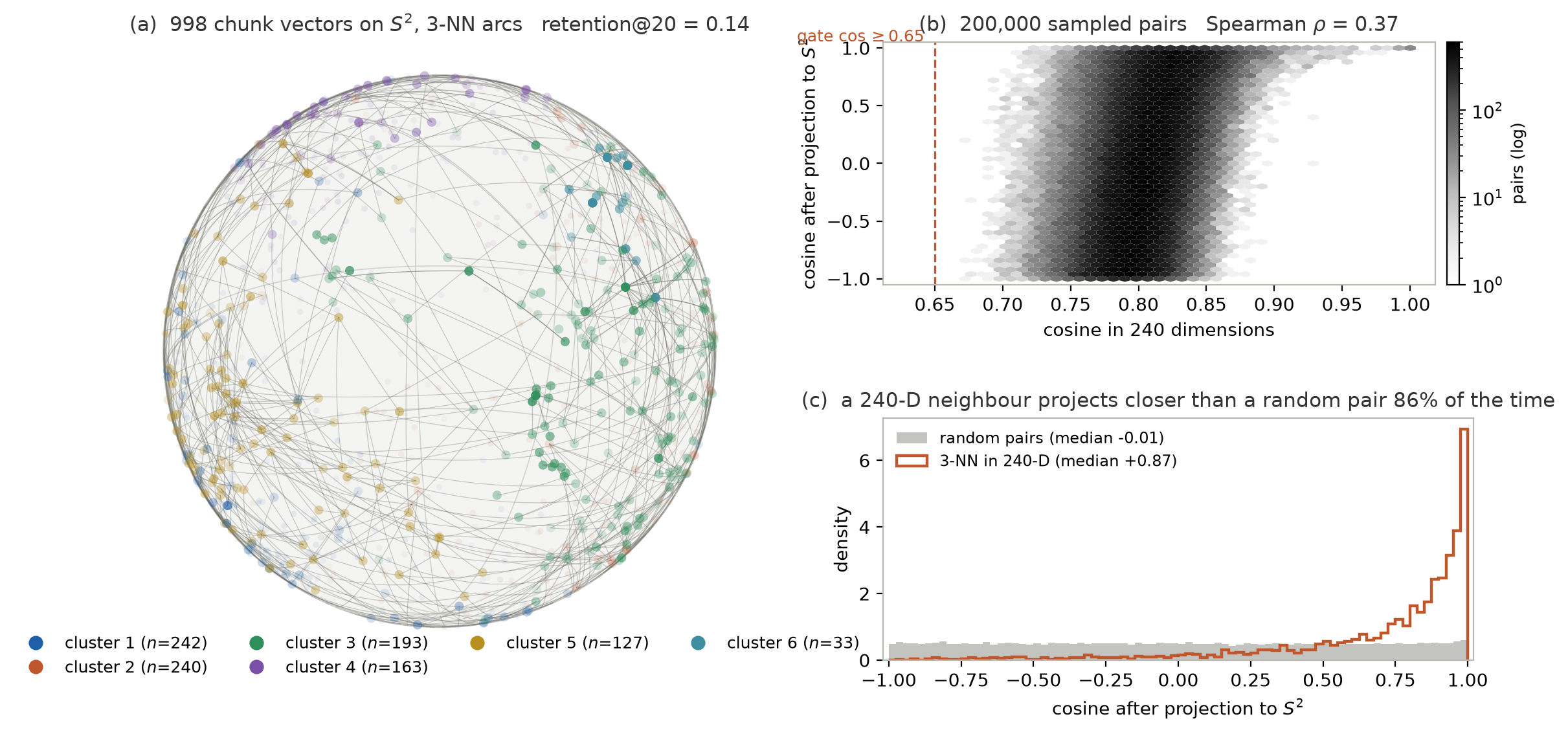}
\caption{The corpus embedded on $S^2$: $998$ chunk vectors from $246$ documents.
\textbf{(a)} chunk vectors projected onto
the three leading principal components, coloured by a spherical $k$-means
computed in the \emph{full} $240$ dimensions --- so the colours are a fact about
the embedding, and the picture either shows them holding together or it does
not. Each grey arc joins a chunk to one of its three nearest neighbours in
$240$ dimensions; short arcs are neighbourhoods the projection kept, arcs across
the sphere are neighbourhoods it lost. Point opacity is $r_i$, the fraction of
that point's own deviation from the corpus mean that survived the projection, so
a washed-out point is one the figure is least entitled to make a claim about.
\textbf{(b)} projected cosine against true $240$-dimensional cosine.
\textbf{(c)} the claim as a distribution rather than a shape: where
$240$-dimensional nearest-neighbour pairs land after projection, against random
pairs.}
\label{fig:sphere}
\end{figure}

Two coordinates fall out of the projection for free and both are reported rather
than hidden. The first is $r_i = \lVert P(x_i-\mu)\rVert / \lVert x_i-\mu\rVert$,
how much of point $i$'s deviation survived; it is drawn as opacity. The second
is the only thing a $k$-NN figure can honestly claim,
\[
  \mathrm{retention}@k \;=\; \frac{1}{n}\sum_i
    \frac{\lvert \mathrm{topk}_{240}(i) \cap \mathrm{topk}_{\mathrm{proj}}(i)\rvert}{k},
\]
the share of each point's true $k$ nearest neighbours that are still among its
$k$ nearest after projection. It belongs in the caption of any such figure, and
its absence elsewhere is usually the tell.

Those numbers also settle the choice of a sphere over the circle the idea
naturally suggests. $S^1$ has one degree of freedom: $n$ points on it are $n$
angles on a line that wraps, so once $n$ reaches the thousands, neighbouring
points sit a fraction of a degree apart and everything looks clustered because
there is nowhere else to be. A picture that cannot fail is not evidence. On the
$998$ chunks of this run the third coordinate is worth $\mathrm{retention}@20$ of
$0.142$ against $0.080$ for the circle, and $\mathrm{retention}@50$ of $0.223$
against $0.145$ --- $1.8\times$ and $1.5\times$ for one extra number per point.

Panel (c) is the result, and it is worth stating in the order that makes it
falsifiable. The projection is \emph{bad} at what a projection is usually asked
to do: three axes keep $11.0\%$ of the variance, the median point retains
$0.304$ of its own deviation from the corpus mean, only $14.2\%$ of a point's
true $20$ nearest neighbours are still among its nearest $20$ afterwards, and
the rank correlation with the true cosine is a weak $\rho = 0.371$. Every one of
those numbers got worse as the corpus grew --- at $56$ chunks the same
projection kept $21.5\%$ of the variance at $\rho = 0.478$ --- which is what
crowding $240$ dimensions into three must do and is exactly why the figure
reports them.

And yet the separation survives all of it. A pair that are nearest neighbours in
$240$ dimensions project to a median cosine of $+0.865$, against $-0.007$ for a
random pair, and a neighbour lands closer than a random pair $85.6\%$ of the
time --- at $998$ chunks, essentially unchanged from $85.3\%$ at $56$. That gap,
not the shape of panel (a), is what ``similar documents cluster spatially''
means operationally: the projection has lost which neighbour is which while
keeping that neighbours are near. It is also the cleanest argument for why the
pass queries the full $240$ dimensions and never a projection --- $11\%$ of the
variance is enough to \emph{see} the neighbourhood structure and nowhere near
enough to \emph{compute} on it.

Two things are deliberately not drawn. The admission gate $\beta = 0.35$ is an
angle of $49.46^{\circ}$ in $240$ dimensions; drawing it as a cap on the
projected sphere would assert that the projection preserves angles, which it
does not. Nor are the gate-admitted pairs drawn, which was the first attempt: on
this corpus the gate admits $497{,}499$ of $497{,}503$ chunk pairs --- all but
four --- because every document is about the same subject, so ``admitted'' is
very nearly the complete graph and an arc for each says nothing. That is itself a result --- the gate is a cheap
pre-filter against unrelated text, not a discriminator within a topical corpus,
and the discrimination is done downstream by $\mathit{min\_weight}$.

%% file: sections/10-limitations.tex
\section{Limitations}
\label{sec:limitations}

Several of these are the direct cost of a decision defended earlier, and a reader
deciding whether to adopt the method needs both halves.

\paragraph{No precision or recall.} The central omission. We do not know what
fraction of the \EMB{} edges a human annotator would endorse, nor what fraction
of the endorsable relationships the pass recovers. Both need a corpus with hidden
relationships labelled by someone who read it, and no such corpus exists here.
The agreement figures of Table~\ref{tab:dims} compare two configurations of one
method; reading them as accuracy would be a misuse. The obvious study ---
annotate a few hundred candidate pairs, blind, and report precision at several
thresholds --- is the natural next step and is not attempted here.

\paragraph{The value term is analysed, not validated.} Section~\ref{sec:idw}
shows that the affine term is unusable behind a gate and that the $p=1$ weight
flattens the dynamic range. Both are arguments from the shape of the formulas,
supported by Table~\ref{tab:spread}. Neither shows that the chord term produces
\emph{better edges} than some third alternative; that comparison needs the ground
truth of the previous paragraph.

\paragraph{On this corpus the gate and the threshold barely bind.} Over all
$509{,}545$ chunk pairs the cosine runs from $0.630$ to $1.000$, with a median of
$0.805$ and an interquartile range only $0.038$ wide. The gate
$\cos \geq 1-\beta = 0.65$ therefore admits $100.00\%$ of them, and downstream
$\theta = 0.7$ passes $1{,}390{,}600$ of $1{,}415{,}574$ node pairs, or $98.2\%$.
The per-endpoint cap does essentially all of the selecting. The analysis is not
invalidated --- an embedder that concentrates distances more tightly than the
gate makes the collapse of Proposition~\ref{prop:affine} worse, not better --- but
on a single-domain corpus at $240$ dimensions the two parameters a user is offered
are not the ones controlling the result. Whether that is a property of the domain,
the embedder, or truncation to $240$ dimensions we have not separated; the corpus
has no $768$-dimensional counterpart to compare against.

\paragraph{$\varepsilon$ is a decision, not a guard, and it is untuned.} At
$\varepsilon = 10^{-6}$ a same-chunk co-occurrence outweighs a neighbour at
$\cos = 0.9$ by about $10^5$. Section~\ref{sec:idw} argues that same-chunk
evidence \emph{should} dominate, and the ordering is surely right, but the
magnitude was inherited from a numerical-guard idiom rather than chosen. A corpus
of long paragraphs, where sharing one is weak evidence, would want a much larger
$\varepsilon$, and we have not characterised where that crossover lies.

\paragraph{Changing the value term invalidates stored accumulators, silently.}
Accumulators written under one $v$ cannot be mixed with those written under
another, and nothing detects the mixture (Section~\ref{sec:incremental}). The
remedy in use is procedural --- re-extract into a fresh graph --- which is not a
remedy.

\paragraph{The per-endpoint cap does not bind at corpus scale.} It is a per-run
bound, so a node accumulates a fresh top $n$ from every document that mentions
it. Section~\ref{subsec:derived} quantifies the result: with $n = 10$ over $246$
documents, $55.3\%$ of endpoints finish above the nominal cap and the largest
derived degree is $742$. The trade is deliberate, but $n$ cannot be read as a
degree bound in any useful sense.

\paragraph{The derived layer dominates the graph, and the query side is not yet
adapted to that.} At the end of the corpus run $89.6\%$ of all edges are \EMB{}.
That the ratio is expected does not dispose of it: a traversal that ignores
weights will find a derived path between almost any two nodes, and will find it
more easily than the extracted path that actually states something. Keeping
\EMB{} out of the natural-language answer path (Section~\ref{sec:impl}) protects
that one consumer and leaves every other to discover the problem for itself. A
weight-aware traversal primitive, or a materialised view at a higher threshold,
is the missing piece.

\paragraph{Facet names are not capped, so the node-label count does not
converge.} Canonical \emph{types} saturate at $24$ over $246$ documents, but the
graph ends with $101$ distinct node labels, because a node carries its facets as
labels too and the cap counts canonicals per kind rather than labels per node
(Section~\ref{subsec:vocab}). Some of the excess are not facets of anything ---
\code{Correct}, \code{Different}, \code{Frustrating} --- but adjectives lifted out
of prose. This is upstream of the similarity pass. It matters here because
``the ontology converges'' has been shown of the structural vocabulary only.

\paragraph{Exact $k$-NN is linear in corpus size per document.} The search side
of the join is the whole vector table (Algorithm~\ref{alg:master}, step~8). At
today's scale the constant is small and exactness is worth more than the saving,
but the reasoning has a horizon, and Section~\ref{subsec:cost} measures where it
falls: not per document, where the scan stays negligible against extraction, but
in the cumulative ingest, around ten million chunks. An HNSW
index~\citep{malkov2020,douze2024} is built past a corpus-size threshold and
takes over there, with candidates rescored on the exact metric so the written
weights are never approximate. The residual limitation is narrower than it was:
what remains approximate is only \emph{which} pairs are offered for scoring, and
that was measured at $99.6$--$100\%$ edge-level agreement rather than assumed.
The trade-off has not been eliminated, only bounded on one corpus.

\paragraph{Warehouse-backed graphs are excluded.} Where a label is part of the
schema, there is no additive element write and the link stage cannot run
(Section~\ref{sec:impl}). The pass declines by name rather than failing, but the
exclusion is real: there the layer would have to be materialised as its own
relation behind a view cascade, which is a different design.

\paragraph{Whole-document extraction is incompatible.} With one chunk per
document every entity co-occurs with every other at $\cos = 1$ and the pass emits
a clique. Forcing chunked extraction resolves the failure and also means the two
modes cannot be combined.

\paragraph{Edges are never retracted.} The pass has no delete path, so an \EMB{}
edge whose weight later falls below $\theta$ --- possible, since the weight is a
mean --- stays in the graph. Raising $\theta$ after the fact is a different
operation and is available: every written edge carries its weight, so tightening
the threshold on an existing graph is one \code{DELETE} predicated on
\code{r.weight} rather than a re-run. On the corpus graph that is worth doing,
but it is an operator action outside the pass.

\paragraph{Chunking is paragraph-based and capped.} A document with more
paragraphs than the cap is truncated. The truncation is reported rather than
silent, but the tail is not embedded and its entities acquire no similarity
evidence.

\paragraph{One embedding provider.} The pass supports a single embedder family
and refuses to run when the configured chat provider is not the matching one,
rather than handing another provider's credentials to the embedding client. That
is the right failure, but it ties the layer's availability to a provider choice
made elsewhere.

\paragraph{Single-language and single-modality.} Everything here assumes text
paragraphs in a language the embedder handles well. Cross-lingual behaviour is
untested, and non-text content contributes nothing.

%% file: sections/11-conclusion.tex
\section{Conclusion}
\label{sec:conclusion}

An extractor that only asserts what a sentence says will build a graph that is
correct and disconnected. The pass described here recovers the ties the corpus
implies without weakening that guarantee. It keeps the derived layer strictly
separate: one edge label carrying a weight and a support count, written through
the seam extraction already uses, with no path by which it can alter or remove
anything extraction found.

Two things in the design turned out to matter more than expected.

The first is that transplanting inverse-distance weighting into a cosine space is
not merely a change of metric. On $L_2$-normalised vectors the weight
$1/(\varepsilon + 1 - \cos)$ is exactly Shepard $p=2$, so the exponent is fixed
rather than free. But the value term is a function of the same distance that
produces the weight, which has no counterpart in the spatial setting, and that is
where the design goes wrong if it is going to. Behind a gate admitting
$\cos \geq 1-\beta$, the natural affine value term is confined to a band of width
$\beta/2$ and renders any write threshold inoperative. The rescaled chord
distance does not collapse, and it comes with an explicit constraint,
$\theta > 1 - \sqrt{\beta/2}$, tying the threshold to the gate. Neither the
failure nor the constraint is visible from the formulas until the gate is taken
into account, and the failure is silent: a threshold that filters nothing looks
exactly like one set generously.

The second is that persisting the accumulators \emph{before} thresholding buys
almost everything one wants from an incremental system, and costs one
representational decision rather than any machinery. The accumulator triple is a
commutative monoid, so arrival order cannot affect the final weights, nothing has
to be recomputed as the corpus grows, and a pair whose evidence is currently too
thin keeps accumulating until a later document tips it over the threshold. The
layer evolves with the corpus because addition is associative --- not because
anything watches for changes.

What is missing is the evaluation. We can say what the method computes, why the
formulation is the one it is, and what it costs; we cannot yet say how often the
edges it produces are ones a reader would endorse. Building a labelled set of
hidden relationships and reporting precision across thresholds is the necessary
next work, and it would also make the open parameters --- $\varepsilon$ above all,
which sets how much a shared paragraph outweighs a similar one and is currently
inherited rather than chosen --- tunable against something other than judgement.
Beyond that, the two clear extensions are an approximate neighbour index, which
the exact self-join will require once the corpus reaches the order of millions of
chunks, and the combination of this text-provenance signal with a structural
link-prediction signal, which draws on entirely different evidence and should be
complementary rather than redundant.